\documentclass[letterpaper, 10 pt, conference]{ieeeconf}  

\IEEEoverridecommandlockouts                              

\usepackage{amsmath,amssymb,amsfonts}
\usepackage{graphicx}
\usepackage{cite}
\usepackage{xcolor}
\usepackage{verbatim}
\usepackage{algorithm}
\usepackage{algpseudocode}
\usepackage{subcaption}

\makeatletter
\let\NAT@parse\undefined
\makeatother

\usepackage{hyperref}

\newtheorem{theorem}{Theorem}[section]
\newtheorem{lemma}[theorem]{Lemma}
\newtheorem{proposition}[theorem]{Proposition}

\newtheorem{definition}[theorem]{Definition}
\newtheorem{remark}[theorem]{Remark}
\newtheorem{example}[theorem]{Example}
\newenvironment{proofsketch}
{\par\vspace{0.3em}\noindent\hspace{2em}{\itshape Proof sketch: }\ignorespaces}
{\hspace*{\fill}~\QED\par}

\title{\LARGE \bf
On a Gradation for Asymptotic Stability
}

\author{Yigit Narter and Inigo Incer%
\thanks{Department of Electrical and Computer Engineering, University of Michigan, 1301 Beal Avenue, Ann Arbor, MI 48109-2122, USA.
Emails: {\tt\small \{narter,iir\}@umich.edu}. Corresponding author: Y. Narter.}%
}

\begin{document}

\maketitle
\thispagestyle{empty}
\pagestyle{empty}

\begin{abstract}
Classical asymptotic stability guarantees convergence but does not quantify
the rate at which convergence occurs. This paper introduces a gradation of
asymptotic stability where degree zero corresponds to exponential stability
and degree \(m>0\) corresponds to algebraic decay of order \(t^{-1/m}\).
We provide direct and converse Lyapunov characterizations for admissible degrees and
conditions for certifying the exact stability degree. Hopf, Bautin,
fractional-degree, and time-varying examples demonstrate how the degree
identifies the leading stabilizing mechanism.
\end{abstract}


\section{Introduction}

The stability of dynamical systems is a central topic in nonlinear analysis,
control theory, and applied mathematics
\cite{lyapunov1992general,hahn1967stability,khalil2002nonlinear}.
The key notion of Lyapunov stability requires trajectories of a dynamical
system to stay arbitrarily close to an equilibrium point if they start
sufficiently close to it. This notion is further qualified as asymptotic
stability when trajectories converge to the equilibrium point, and as
exponential stability when this convergence satisfies an exponential bound.
There is, however, a spectrum of bounds between convergence and exponential
convergence.

This gap is visible in the classical Hopf bifurcation, where a conjugate
pair of eigenvalues crosses the imaginary axis and a small-amplitude
periodic orbit may emerge or disappear. The usual
bifurcation-theoretic classification emphasizes the qualitative phase
portrait: trajectories may spiral into an equilibrium, spiral away from it,
or approach a stable or unstable limit cycle, depending on the signs of the
unfolding parameter and the first Lyapunov coefficient
\cite{guckenheimer2013nonlinear,kuznetsov2004one}.
However, this classification does not distinguish between different
convergence laws within the attracting equilibrium regime.

Indeed, near a Hopf bifurcation, the leading radial dynamics can be written in
normal-form coordinates as
\(\dot r=r(\lambda+\alpha r^2)\), with angular dynamics decoupled to leading
order. When \(\lambda<0\), the linear radial term is stabilizing and
trajectories converge exponentially to the origin. At the critical value
\(\lambda=0\), however, the linear radial damping vanishes. If \(\alpha<0\),
the origin is still locally attracting, but the first stabilizing term is now
cubic: \(\dot r=-|\alpha|r^3\). Consequently,
\(r(t)\asymp(t-t_0)^{-1/2}\) as \(t\to\infty\), so the equilibrium remains
asymptotically stable while its convergence law changes from exponential to
algebraic. Here and throughout, \(f(t)=O(g(t))\) as \(t\to\infty\) means that
\(|f(t)|\le C|g(t)|\) for all sufficiently large \(t\) and some \(C>0\),
while \(f(t)\asymp g(t)\) means that both \(f(t)=O(g(t))\) and
\(g(t)=O(f(t))\). Thus two regimes that are both ``spiraling into the origin''
from the qualitative phase-portrait viewpoint have different decay mechanisms
and different rates.

An even sharper separation occurs in generalized Hopf, or Bautin,
bifurcations. A Bautin point is a degenerate Hopf bifurcation at which the first Lyapunov
coefficient vanishes, so that the cubic radial term disappears and the
quintic term becomes relevant. Near such a point, the radial normal form can be written as \(\dot r=r(\beta_1+\beta_2 r^2-a r^4)\) for \(a>0\), again with angular
dynamics decoupled to leading order \cite{Guckenheimer:2007,yu2007simplest}. In this family,
\(\beta_1<0\) results in exponential decay; \(\beta_1=0\), \(\beta_2<0\)
gives \(r(t)\asymp(t-t_0)^{-1/2}\); and
\(\beta_1=\beta_2=0\) yields
\(r(t)\asymp(t-t_0)^{-1/4}\), as \(t\to\infty\).
Thus equilibria that are all asymptotically stable can exhibit substantially
different convergence laws depending on which stabilizing coefficient is the
first nonzero one. Figure~\ref{fig:hopf-bautin-decay} illustrates this rate
separation.

\begin{figure}[t]
    \centering
    \begin{subfigure}{0.48\linewidth}
        \centering
        \includegraphics[width=\linewidth]{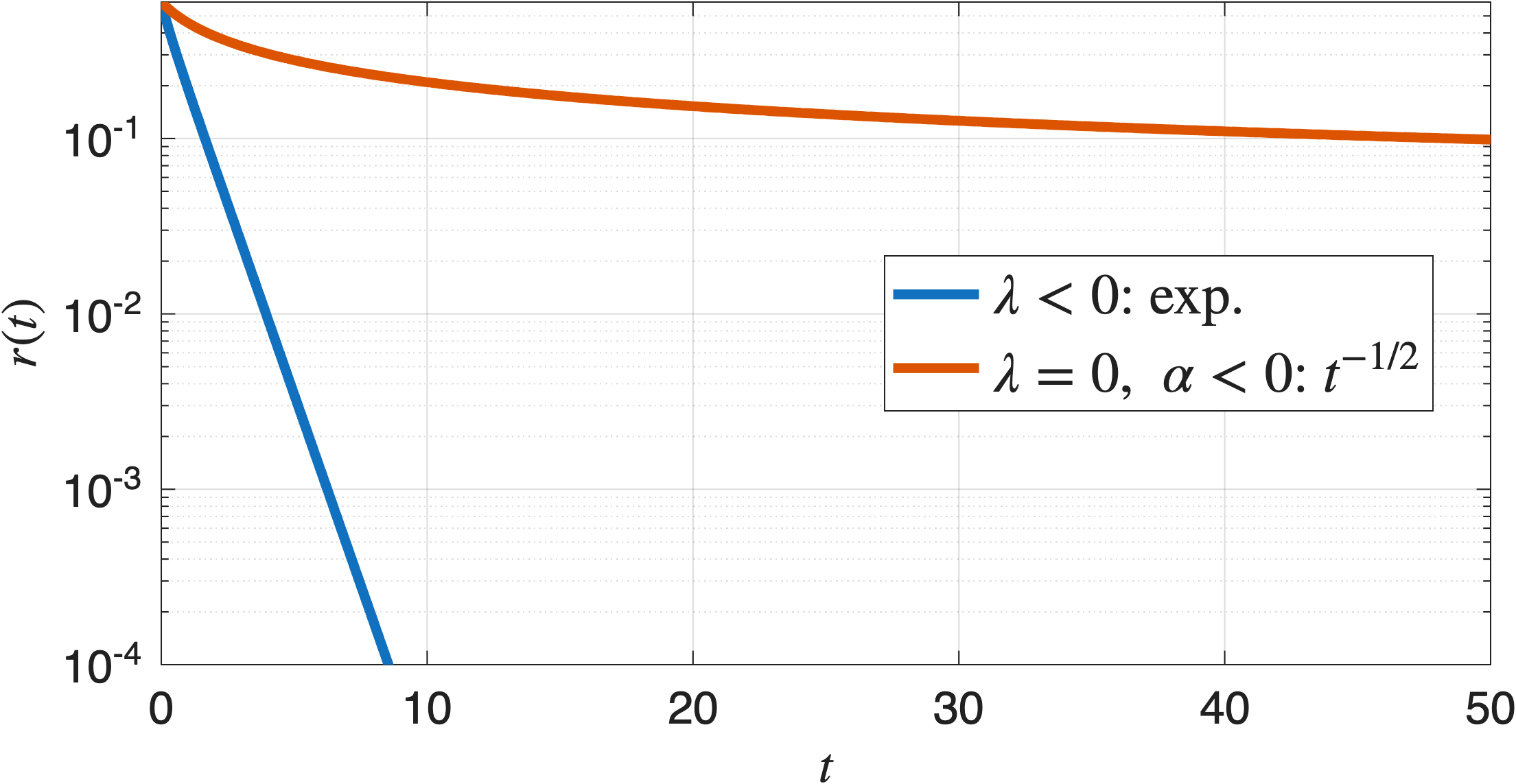}
        \caption{Hopf radial decay.}
        \label{fig:hopf-decay}
    \end{subfigure}
    \hfill
    \begin{subfigure}{0.48\linewidth}
        \centering
        \includegraphics[width=\linewidth]{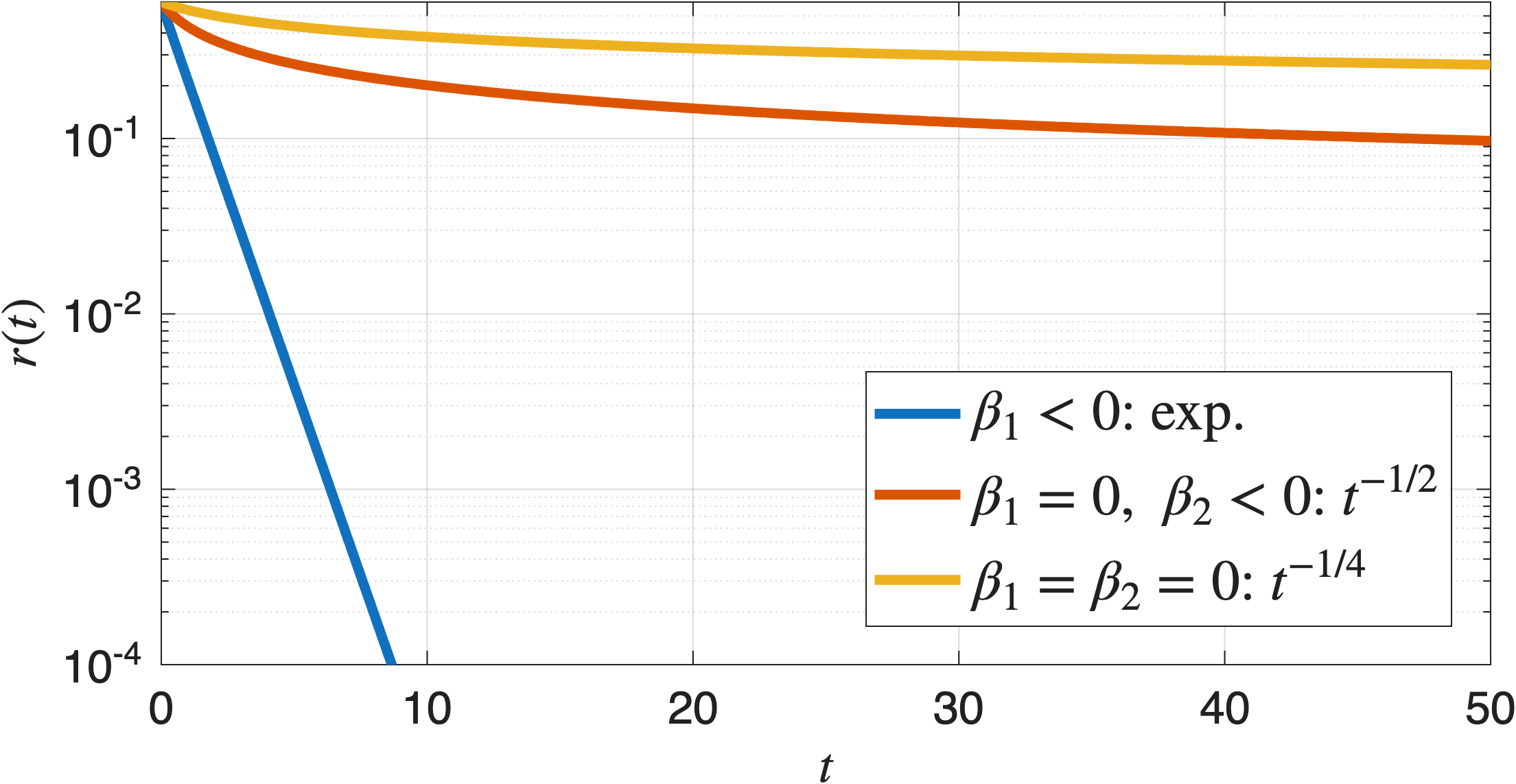}
        \caption{Bautin radial decay.}
        \label{fig:bautin-decay}
    \end{subfigure}
    \caption{Radial decay for Hopf and Bautin: in our framework, degree \(0\) corresponds to exponential decay, while degrees \(2\) and \(4\) correspond to \(t^{-1/2}\) and \(t^{-1/4}\) decay, respectively.}
    \label{fig:hopf-bautin-decay}
\end{figure}

These examples show that the decay rate is not merely a quantitative detail,
but a structural property of the vector field near the equilibrium. Classical
asymptotic stability groups all of the above stable cases under the same
label, while exponential stability only distinguishes the case in which a
linear stabilizing term is present. The intermediate algebraic regimes are
therefore not explicitly organized by the usual qualitative terminology.

\textbf{Contributions.}
Motivated by this gap, we introduce a rate-based degree framework for local
asymptotic stability. The framework is built from a canonical family of
comparison envelopes: degree \(0\) corresponds to exponential convergence,
while degree \(m>0\) corresponds to an admissible algebraic decay bound of
order \(O(t^{-1/m})\). We prove that degree estimates imply asymptotic
stability and develop Lyapunov certificates for this gradation. In particular,
an upper Lyapunov dissipation inequality certifies an admissible degree \(m\),
while a matching lower inequality on a forward-invariant set yields
trajectories of order \(\asymp t^{-1/m}\), rules out every faster degree, and
thereby certifies the exact degree. We also provide a nonsmooth converse
characterization using upper Dini derivatives along solutions: a degree-\(m\)
estimate is equivalent to the existence of a continuous Lyapunov certificate
with matching degree-\(m\) dissipation. Overall, the proposed framework
complements classical Lyapunov stability theory by making decay rates explicit.

\section{Related Work}

Our work builds on classical Lyapunov stability theory and converse Lyapunov
theorems. Uniform asymptotic stability can be characterized through
class-\(\mathcal{KL}\) trajectory estimates and Lyapunov functions; see, e.g.,
\cite{khalil2002nonlinear,sontag2008input,kellett2015classical}. Converse
results establish Lyapunov functions from such trajectory estimates in broad
settings, including time-varying systems
\cite{lin1996smooth,teel2000smooth}. These results provide general
characterizations of asymptotic stability, whereas our objective is to
introduce a one-parameter gradation that distinguishes asymptotically stable
equilibria by their convergence rates.

Power-type Lyapunov inequalities are classical:
\(\dot V\le-cV^\beta\) implies finite-time stability for \(0<\beta<1\)
\cite{bhat2000finite}, exponential decay for \(\beta=1\)
\cite{khalil2002nonlinear}, and algebraic decay for \(\beta>1\)
\cite{caraballo2001decay,jammazi2013rational}.
Rather than fixing the exponent, our framework treats the corresponding
decay order as a stability parameter and identifies its smallest admissible
value. Closely related is rational stability
\cite{bacciotti2005liapunov,jammazi2013rational}, commonly expressed through
an estimate of the form
\begin{equation}
\label{eq:rational}
    \|\psi(t,x)\|^p\le M\|x\|^p/(1+k\|x\|^p t),
\end{equation}
where \(M,k>0\) and \(p\)
is generally taken to be a positive integer. Here, \(p\) is fixed a priori and specifies the algebraic order \(O(t^{-1/p})\). Our framework instead treats the corresponding order \(m\) as the stability
parameter to be determined: degree \(m=0\) denotes
exponential stability, while a real degree \(m>0\) corresponds to an
admissible decay envelope of order \(t^{-1/m}\). Without prescribing this order a priori, we consider the set of all admissible orders and define its smallest element, when attained, as the degree of the equilibrium, with matching Lyapunov lower bounds certifying exactness and thereby organizing exponential and algebraic convergence within a
single gradation.

The recent work of Jagt and Peet~\cite{jagt2026lyapunov} is particularly
close to ours. They develop a framework for quantifying the rate performance
of autonomous nonlinear ODEs relative to a prescribed family of trajectory
bounds. In the rational case, this again takes the form \eqref{eq:rational} after the decay model and
state exponent \(p\) are specified; \(k\) measures the rate and \(M\) the
gain. More generally, their framework fixes a state measure and a comparison
function and derives necessary and sufficient Lyapunov conditions for the
associated rate and gain performance. Our framework addresses a different
question: instead of fixing the algebraic decay order and quantifying the
constants in the resulting estimate, we make that order itself the quantity
to be classified. The single parameter \(m\) determines the family
\(t^{-1/m}\), and the smallest admissible \(m\), when attained, determines the degree of
stability. Accordingly, our objective is not to optimize rate or gain constants for a prescribed decay family, but to determine the boundary of admissible decay orders and certify that boundary as the equilibrium's
stability degree; see Example~\ref{ex:hopf} for a
direct comparison. Moreover, our framework and converse characterization apply directly to time-varying
systems \(\dot x=f(t,x)\), whereas \cite{jagt2026lyapunov} develops its rate-performance framework for autonomous ODEs.

Finally, the proposed degree is distinct from the classical Lyapunov exponent
used to measure exponential growth or decay
\cite{colonius1999lyapunov}. Since
\(t^{-1}\log\|x(t)\|\to0\) for algebraically decaying trajectories, different
algebraic rates have the same zero Lyapunov exponent. Our degree refines this
regime by distinguishing their algebraic decay orders.

\section{Degree of Stability}

\newcommand{\reals}[0]{\mathbb{R}}
\newcommand{\nnreals}[0]{\reals_{\ge 0}}
\newcommand{\preals}[0]{\reals_{> 0}}

Consider the system
\begin{equation}
    \dot x = f(t,x), \qquad x\in\mathbb R^n,
    \label{eq:system}
\end{equation}
where \(f(t,0)=0\) for all \(t\ge0\), and \(f\) is continuous in \(t\)
and locally Lipschitz in \(x\), uniformly with respect to \(t\) on compact
sets. These assumptions ensure local existence and uniqueness of solutions, which
we denote by \(x(t,t_0,x_0)\). Throughout the paper, we study the equilibrium
at the origin; a general equilibrium can be translated to the origin by a
change of coordinates.

As mentioned in the previous sections, classical asymptotic stability specifies
that trajectories converge to the equilibrium, but it does not distinguish
between qualitatively different rates of convergence. In particular, an
exponentially stable equilibrium and an equilibrium whose trajectories decay
only algebraically are both classified as asymptotically stable. We refine this
distinction by introducing a family of canonical decay envelopes. The parameter
\(m \in \nnreals\) indexes the rate of convergence: \(m=0\) corresponds to exponential
decay, while \(m>0\) represents the algebraic order \(t^{-1/m}\).

For \(\alpha \in \preals\) and \(t,\rho \in \nnreals\), define
\begin{equation}
\label{eq:degree-m-envelope}
\phi_m^\alpha(t;\rho):=
\begin{cases}
\rho e^{-\alpha t}, & m=0, \\[4pt]
\displaystyle
\rho\left(1+m\alpha\rho^m t\right)^{-1/m},
& m>0.
\end{cases}
\end{equation}
Note that, for each \(m\ge0\), the function \(\phi_m^\alpha\) is continuous in \((t,\rho)\). Moreover, the definition is right continuous at \(m=0\) for a fixed \((t,\rho)\), since \(\lim_{m\to0^+}\rho(1+m\alpha\rho^m t)^{-1/m}=\rho e^{-\alpha t}\). For \(m>0\), the comparison function is motivated by
\(\dot z=-\alpha z^{m+1}\), whose positive solutions satisfy
\(z(t)=z_0(1+m\alpha z_0^m t)^{-1/m}\). Thus, the family
\(\phi_m^\alpha\) interpolates between exponential and polynomial convergence
rates.

\begin{definition}[Admissible degree-\(m\) estimate]
\label{def:degree-m-estimate}
Consider system~\eqref{eq:system} and let \(m\in\mathbb R_{\ge0}\).
We say that the origin admits a \emph{degree-\(m\) estimate} if, for each
\(t_0\ge0\), there exist constants
\(\alpha_{t_0},c_{t_0},r_{t_0}>0\) such that for every
\(x_0\in B_{r_{t_0}}(0)\) and every \(t\ge t_0\), the corresponding solution
satisfies
\begin{equation}
\label{eq:degree-m-estimate}
    \|x(t,t_0,x_0)\|
    \le
    c_{t_0}\,
    \phi_m^{\alpha_{t_0}}(t-t_0;\|x_0\|).
\end{equation}
The estimate is called \emph{uniform} if the constants
\(\alpha_{t_0},c_{t_0},r_{t_0}>0\) can be chosen independently of \(t_0\).
\end{definition}

The next result shows that uniform degree estimates are not merely rate
conditions, they already imply local uniform asymptotic stability. Thus degree estimates refine asymptotic stability into subclasses indexed by convergence rate.

\begin{proposition}
\label{prop:degree-estimate-implies-uas}
For system~\eqref{eq:system}, if the origin admits a uniform degree-\(m\) estimate for some
\(m\in\mathbb R_{\ge0}\), then the origin is locally uniformly
asymptotically stable.
\end{proposition}

\begin{proof}
By Definition~\ref{def:degree-m-estimate} and~\eqref{eq:degree-m-estimate}, the uniform degree-\(m\) estimate
provides constants \(\alpha,c,r>0\), independent of \(t_0\), such that
\(\|x(t,t_0,x_0)\|\le
c\,\phi_m^\alpha(t-t_0;\|x_0\|)\)
for every \(t_0\ge0\), \(x_0\in B_r(0)\), and \(t\ge t_0\).
Since \(m\ge0\), \(\alpha>0\), and \(t-t_0\ge0\), \eqref{eq:degree-m-envelope} gives
\(\phi_m^\alpha(t-t_0;\|x_0\|)\le\|x_0\|\).
Hence \(\|x(t,t_0,x_0)\|\le c\|x_0\|\) for all \(t\ge t_0\).
Given \(\varepsilon>0\), choose
\(\delta:=\min\{r,\varepsilon/c\}\).
Then \(\|x_0\|<\delta\) implies
\(\|x(t,t_0,x_0)\|<\varepsilon\) for all \(t\ge t_0\),
uniformly in \(t_0\). Thus the origin is locally uniformly stable.

It remains to prove local uniform attractivity. Fix any
\(\delta_a\in(0,r)\). Since \(\phi_m^\alpha(s;\rho)\) is nondecreasing in
\(\rho\), for every \(x_0\in B_{\delta_a}(0)\) and \(t\ge t_0\) we have
\(\|x(t,t_0,x_0)\|\le
c\,\phi_m^\alpha(t-t_0;\delta_a)\).
Moreover, by~\eqref{eq:degree-m-envelope}, \(\phi_m^\alpha(s;\delta_a)\to0\) as \(s\to\infty\) for every
\(m\ge0\) and \(\alpha>0\). Hence, for every \(\varepsilon>0\), there exists
\(T>0\), independent of \(t_0\) and \(x_0\in B_{\delta_a}(0)\), such that
\(\|x(t,t_0,x_0)\|<\varepsilon\) whenever \(t\ge t_0+T\).
Thus the origin is locally uniformly attractive. Together with local uniform
stability, this proves local uniform asymptotic stability.
\end{proof}

We now record the rate interpretation of the parameter \(m\). This result
shows that the comparison family \(\{\phi_m^\alpha:m>0\}\) represents all
algebraic decay bounds.

\begin{proposition}
\label{prop:degree-parameter-algebraic-decay}
For system~\eqref{eq:system}, let \(m>0\). If the origin admits a degree-\(m\) estimate, then, for every fixed \(t_0\ge0\) and every \(x_0\in B_{r_{t_0}}(0)\), \(\|x(t,t_0,x_0)\|=O\!\left((t-t_0)^{-1/m}\right)\) as \(t\to\infty\), where \(r_{t_0}\) is as in Definition~\ref{def:degree-m-estimate}.
\end{proposition}

\begin{proof}
Fix \(t_0\ge0\) and let the corresponding constants in~\eqref{eq:degree-m-estimate} be
\(\alpha_{t_0},c_{t_0},r_{t_0}>0\). If \(x_0=0\), the claim is trivial. If
\(x_0\neq0\), then
\(\|x(t,t_0,x_0)\|\le c_{t_0}\|x_0\|(1+m\alpha_{t_0}\|x_0\|^m(t-t_0))^{-1/m}\),
and the right-hand side is asymptotic to
\(c_{t_0}(m\alpha_{t_0})^{-1/m}(t-t_0)^{-1/m}\). Hence
\(\|x(t,t_0,x_0)\|=O((t-t_0)^{-1/m})\). Setting \(m=1/\beta\) gives
the stated correspondence with algebraic decay bounds \(t^{-\beta}\).
\end{proof}

Thus the degree \(m\) corresponds to the algebraic decay exponent
\(\beta=1/m\); equivalently, an algebraic decay bound of order \(t^{-\beta}\) corresponds to \(m=1/\beta\), with larger \(m\) indicating slower convergence. Having
established this rate interpretation, we now define a unique degree associated with an equilibrium by considering the smallest admissible rate index.

\begin{definition}[Lower admissible degree]
\label{def:lower-admissible-degree}
Let \(\mathcal D\subseteq\mathbb R_{\ge0}\) denote the set of admissible degrees,
that is, \(\mathcal D:=\{q\in\mathbb R_{\ge0}:\text{the origin admits a
degree-}q\text{ estimate}\}\). If \(\mathcal D\neq\emptyset\), the
\emph{lower admissible degree} of the origin is
\(\underline m:=\inf\mathcal D\).

Similarly, let \(\mathcal D_u\subseteq\mathbb R_{\ge0}\) denote the set of uniform
admissible degrees, that is,
\(\mathcal D_u:=\{q\in\mathbb R_{\ge0}:\text{the origin admits a uniform
degree-}q\text{ estimate}\}\). If \(\mathcal D_u\neq\emptyset\), the
\emph{uniform lower admissible degree} of the origin is
\(\underline m_u:=\inf\mathcal D_u\).
\end{definition}

Since every uniform estimate is also an estimate,
\(\mathcal D_u\subseteq\mathcal D\). Hence, whenever both sets are
nonempty, \(\underline m\le\underline m_u\). Moreover, for \(q>m\),
\(m\in\mathcal D\) implies \(q\in\mathcal D\), and
\(m\in\mathcal D_u\) implies \(q\in\mathcal D_u\). Indeed, on
\(0\le\rho\le r\), choosing
\(\alpha_q=m\alpha/(q r^{q-m})\) for \(m>0\) gives
\(\phi_m^\alpha(t;\rho)\le\phi_q^{\alpha_q}(t;\rho)\), for \(m=0\), exponential decay is bounded by every algebraic envelope.
Thus, admissibility of degree \(m\) implies admissibility of every
degree \(q>m\).

\begin{definition}[Degree-\(m\) stability]
\label{def:vector-degree-m-stability}
Suppose that \(\mathcal D\neq\emptyset\). The origin is said to be
\emph{degree-\(m\) stable} if \(m=\underline m\) and \(m\in\mathcal D\).
Equivalently, the origin is degree-\(m\) stable if \(m=\min\mathcal D\).

Similarly, suppose that \(\mathcal D_u\neq\emptyset\). The origin is said to be
\emph{uniformly degree-\(m\) stable} if \(m=\underline m_u\) and
\(m\in\mathcal D_u\). Equivalently, the origin is uniformly degree-\(m\)
stable if \(m=\min\mathcal D_u\).
\end{definition}

\begin{remark}
\label{rem:nonattainment-infinite-degree}
The lower admissible degree \(\underline m=\inf\mathcal D\) is well-defined whenever \(\mathcal D\neq\emptyset\), since \(\mathcal D\subseteq\mathbb R_{\ge0}\) is bounded below by zero, but the infimum need not be attained. For example, a rate \(\|x(t)\|\asymp(\log t)t^{-1/m_*}\), \(m_*>0\), is not bounded by a
degree-\(m_*\) envelope but is bounded by degree-\(q\) envelopes for every
\(q>m_*\), giving \(\mathcal D=(m_*,\infty)\).

This finite nonattainment is distinct from the case of an asymptotically
stable origin with \(\mathcal D=\emptyset\), for which no finite degree
estimate holds. In this case, one may formally assign \(\underline m=\infty\) to represent convergence slower than every algebraic rate, such as
\(1/\log t\). For an origin that is not asymptotically stable, no stability
degree is assigned.
\end{remark}

\begin{remark}
\label{rem:smooth-scalar-finite-leading-order}
Suppose a smooth scalar autonomous system satisfies
\(\dot x=\alpha x^k+o(x^k)\), \(\alpha\neq0\), near a two-sided locally
asymptotically stable origin. Then \(k\) is odd and \(\alpha<0\). If \(k=1\),
the origin has degree \(0\); if \(k>1\), then
\(|x(t)|\asymp t^{-1/(k-1)}\), so
\(\mathcal D=[k-1,\infty)\) and \(\underline m=k-1\in\mathcal D\).
Thus, smooth scalar autonomous systems with a finite nonzero Taylor leading term realize only even integer attained degrees. Flat systems are not covered and may exhibit non-polynomial
or subalgebraic decay. This is consistent with the classical homogeneous
approximation principle~\cite{hahn1967stability}.
\end{remark}



\section{Lyapunov Characterization}
We now turn from trajectory-based definitions to Lyapunov certificates. We first develop a direct Lyapunov method and illustrate it on normal forms, vector systems, and time-varying systems. We then establish a nonsmooth converse characterization of admissible degree-\(m\) estimates.

\subsection{Direct Lyapunov Method for Degree-\texorpdfstring{$m$}{m} Stability}
This subsection presents Lyapunov conditions that certify the trajectory envelopes introduced above. An upper dissipation inequality certifies an admissible degree, while a matching lower inequality on a forward-invariant set rules out faster degrees and certifies the exact degree. The following theorem formalizes this idea.

\begin{theorem}
\label{thm:lyapunov-degree-m}
Consider system~\eqref{eq:system} and let \(m\in\mathbb R_{\ge0}\).
Suppose there exist a continuously differentiable function
\(V:\mathbb R_{\ge0}\times D\to\mathbb R\), where
\(D\subseteq\mathbb R^n\) is a neighborhood of the origin, and constants
\(a,b,\mu,r>0\) such that \(B_r(0)\subseteq D\) and, for all \(t\ge0\) and
\(x\in B_r(0)\),
\begin{equation}
\label{eq:direct-quadratic-bounds}
    a\|x\|^2\le V(t,x)\le b\|x\|^2,
\end{equation}
and
\begin{equation}
\label{eq:direct-upper-dissipation}
    \dot V(t,x)\le -\mu V(t,x)^{1+\frac m2}.
\end{equation}
Then the origin admits a uniform degree-\(m\) estimate as in
Definition~\ref{def:degree-m-estimate}; equivalently,
\(m\in\mathcal D_u\). In particular, if \(m=0\), then \(0=\min\mathcal D_u\), and the origin is uniformly degree-\(0\) stable.

If \(m>0\) and, in addition, there exists a set \(S\subseteq B_r(0)\),
forward invariant in the sense that \(x_0\in S\) implies
\(x(t,t_0,x_0)\in S\) for all \(t_0\ge0\) and \(t\ge t_0\), such that
\(0\in\overline{S\setminus\{0\}}\), together with a constant \(\nu>0\)
such that
\begin{equation}
\label{eq:lyapunov-degree-m-lower}
    \dot V(t,x)\ge -\nu V(t,x)^{1+\frac m2}
\end{equation}
for all \(t\ge0\) and \(x\in S\setminus\{0\}\), then \(m=\min\mathcal D_u\), and the origin is uniformly degree-\(m\) stable as in Definition~\ref{def:vector-degree-m-stability}.
\end{theorem}

\begin{proofsketch}
Choose any \(\bar r\in(0,r)\). By restricting the initial conditions to
\(B_{r_0}(0)\), where
\(r_0=\bar r\sqrt{a/b}\), we have
\(V(t_0,x_0)<br_0^2=a\bar r^2\).
If \(t_e\) were the first exit time from \(B_{\bar r}(0)\), then
\(\dot V\le0\) and \eqref{eq:direct-quadratic-bounds} would give
\(a\bar r^2\le V(t_e,x(t_e))\le V(t_0,x_0)<a\bar r^2\), a contradiction.
Hence \(\|x(t,t_0,x_0)\|<\bar r<r\) and
\(x(t,t_0,x_0)\in B_r(0)\) for all \(t\ge t_0\).
Let \(v(t):=V(t,x(t,t_0,x_0))\). For \(m=0\), the inequality
\(\dot v\le-\mu v\), together with~\eqref{eq:direct-quadratic-bounds}, yields
\(
\|x(t,t_0,x_0)\|
\le
\sqrt{\frac ba}\,
\phi_0^{\mu/2}(t-t_0;\|x_0\|).\) For \(m>0\), comparison with
\(\dot v=-\mu v^{1+m/2}\) gives
\(
v(t)\le v(t_0)
\left(
1+\frac m2\mu v(t_0)^{m/2}(t-t_0)
\right)^{-2/m}.
\)
Using~\eqref{eq:direct-quadratic-bounds},
\(
\|x(t,t_0,x_0)\|
\le
\sqrt{\frac ba}\,
\phi_m^\alpha(t-t_0;\|x_0\|),\)
\(
\alpha=\frac{\mu}{2}a^{m/2}
\)
by~\eqref{eq:degree-m-envelope}. Thus \(m\in\mathcal D_u\). If \(m=0\), then
\(\mathcal D_u\subseteq\mathbb R_{\ge0}\) immediately gives
\(0=\min\mathcal D_u\), proving uniform degree-\(0\) stability.
For \(m>0\), it remains to rule out admissible degrees below \(m\).

Under the additional hypothesis, fix any \(t_0\ge0\). Since
\(0\in\overline{S\setminus\{0\}}\), there are arbitrarily small nonzero
\(x_0\in S\), and forward invariance ensures that the corresponding
trajectories remain in \(S\). Applying~\eqref{eq:lyapunov-degree-m-lower} along such a trajectory and using~\eqref{eq:direct-quadratic-bounds} gives
\(
\|x(t,t_0,x_0)\|
\ge
\sqrt{\frac ab}\,\|x_0\|
\left(
1+\frac m2\nu a^{m/2}\|x_0\|^m(t-t_0)
\right)^{-1/m}.
\)
This is incompatible with any uniform degree-\(q\) estimate~\eqref{eq:degree-m-estimate} with \(q<m\): for \(q=0\), the corresponding upper bound decays exponentially, while
for \(0<q<m\), the ratio between the asymptotic lower and upper rates is
\(
\frac{(t-t_0)^{-1/m}}{(t-t_0)^{-1/q}}
=(t-t_0)^{1/q-1/m}\to\infty,
\)
which is impossible. Hence
\(\mathcal D_u\cap[0,m)=\emptyset\). Since \(m\in\mathcal D_u\), it follows
that \(\underline m_u=m=\min\mathcal D_u\), proving uniform degree-\(m\)
stability. The detailed proof is provided in Appendix~\ref{app:proof-direct}.
\end{proofsketch}

\begin{remark}[Nonuniform version]
\label{rem:nonuniform-lyapunov-version}
Theorem~\ref{thm:lyapunov-degree-m} extends directly to the nonuniform case
by allowing the Lyapunov bounds and dissipation constants to depend on
\(t_0\). The upper inequality~\eqref{eq:direct-upper-dissipation} gives
\(m\in\mathcal D\), while the matching lower
inequality~\eqref{eq:lyapunov-degree-m-lower} on an invariant set accumulating at the origin
yields \(m=\min\mathcal D\) and hence degree-\(m\) stability.
\end{remark}

\begin{remark}
More generally, bounds \(V(t,x)\asymp\|x\|^p\) together with
\(\dot V\le-\mu V^{1+\frac mp}\) yield a degree-\(m\) estimate. For arbitrary
\(p\), however, such bounds may introduce regularity issues, since
\(\|x\|^p\) is smooth at the origin only for even integer \(p\). We therefore
use the standard quadratic scaling \(V(t,x)\asymp\|x\|^2\), for which
\(\dot V\le-\mu V^{1+\frac m2}\) extends the exponential condition
\(\dot V\le-\mu V\), recovered at \(m=0\).
\end{remark}


\subsection{Examples}

We now apply Theorem~\ref{thm:lyapunov-degree-m} to several representative
systems. The first two examples revisit the Hopf-type normal forms from the
introduction and show how the degree changes as lower-order stabilizing terms
vanish. The last two examples illustrate the role of invariant sets in autonomous and time-varying systems.

\begin{example}
\label{ex:hopf}
Consider the Cartesian Hopf normal form
\[
\begin{aligned}
    \dot x_1
    &=
    \lambda x_1-\omega x_2
    +\alpha(x_1^2+x_2^2)x_1
    -\beta(x_1^2+x_2^2)x_2,\\
    \dot x_2
    &=
    \omega x_1+\lambda x_2
    +\beta(x_1^2+x_2^2)x_1
    +\alpha(x_1^2+x_2^2)x_2,
\end{aligned}
\]
where \(\omega\neq0\). Equivalently, in polar coordinates,
\(\dot r=r(\lambda+\alpha r^2)\) and \(\dot\theta=\omega+\beta r^2\). This example represents the standard transition from linear radial damping to
nonlinear radial damping at the critical Hopf parameter. The parameters
\(\omega\) and \(\beta\) affect only the angular motion and do not affect the
radial decay rate.

Let \(V(x)=\frac12(x_1^2+x_2^2)\). Then
\(\dot V=2\lambda V+4\alpha V^2\). If \(\lambda<0\), then
\(\dot V\le -\mu V\) locally for some \(\mu>0\), so
Theorem~\ref{thm:lyapunov-degree-m} gives \(0=\min\mathcal D_u\), and the
origin is uniformly degree-\(0\) stable. If \(\lambda=0\) and
\(\alpha<0\), then \(\dot V=-4|\alpha|V^2\), which satisfies both
\eqref{eq:direct-upper-dissipation} and \eqref{eq:lyapunov-degree-m-lower}
on every sufficiently small positively invariant neighborhood. Hence \(2=\min\mathcal D_u\), and the
origin is uniformly degree-\(2\) stable by Theorem~\ref{thm:lyapunov-degree-m}. Therefore, the Hopf normal form
transitions from exponential degree \(0\) to algebraic degree \(2\),
corresponding to decay \(\asymp(t-t_0)^{-1/2}\), when the linear radial
coefficient vanishes.

This critical case also provides a direct comparison with the rational metric
\eqref{eq:rational} of~\cite{jagt2026lyapunov}: fixing \(p=2\) gives
\(k=2|\alpha|\) and \(M=1\), while the corresponding bound with \(p=1\)
cannot hold. Since \(p\) is prescribed
a priori, this does not identify or certify the minimal admissible order by itself. Our framework, however, excludes every \(q<2\) through the matching lower bound \eqref{eq:lyapunov-degree-m-lower}, yielding \(2=\min\mathcal D_u\) and \(r(t) \asymp (t-t_0)^{-1/2}\).
\end{example}

\begin{example}
\label{ex:bautin}
Consider the generalized Hopf, or Bautin, normal form
\[
\begin{aligned}
    \dot x_1&=-x_2+x_1\bigl(\beta_1+\beta_2(x_1^2+x_2^2)
        -a(x_1^2+x_2^2)^2\bigr),\\
    \dot x_2&=x_1+x_2\bigl(\beta_1+\beta_2(x_1^2+x_2^2)
        -a(x_1^2+x_2^2)^2\bigr),
\end{aligned}
\,\,\, a>0.
\]
Equivalently, \(\dot r=r(\beta_1+\beta_2r^2-ar^4)\) and \(\dot\theta=1\).
This normal form exhibits successive loss of the linear and cubic stabilizing
radial terms.

With \(V(x)=\frac12(x_1^2+x_2^2)\), one has
\(\dot V=2\beta_1V+4\beta_2V^2-8aV^3\). If \(\beta_1<0\), then
\(\dot V\le -\mu V\) locally, so Theorem~\ref{thm:lyapunov-degree-m} gives
\(0=\min\mathcal D_u\), and the origin is uniformly degree-\(0\) stable.
If \(\beta_1=0\) and \(\beta_2<0\), then
\(\dot V=-(4|\beta_2|+8aV)V^2\). On a sufficiently small positively
invariant neighborhood,
\(-C_2V^2\le\dot V\le-C_1V^2\) for some \(C_1,C_2>0\), which satisfies
both \eqref{eq:direct-upper-dissipation} and
\eqref{eq:lyapunov-degree-m-lower}. Hence by
Theorem~\ref{thm:lyapunov-degree-m}, \(2=\min\mathcal D_u\), and the
origin is uniformly degree-\(2\) stable. Finally, if
\(\beta_1=\beta_2=0\), then \(\dot V=-8aV^3\), which satisfies both
inequalities with \(m=4\). Theorem~\ref{thm:lyapunov-degree-m} therefore gives \(4=\min\mathcal D_u\), and the origin is uniformly degree-\(4\) stable. Consequently, the Bautin normal form separates stable cases with degrees
\(0\), \(2\), and \(4\), corresponding respectively to exponential decay
and algebraic decay of orders \(\asymp(t-t_0)^{-1/2}\) and
\(\asymp(t-t_0)^{-1/4}\).
\end{example}

In the next example, the slow dynamics lie on a curved invariant manifold,
which illustrates how the invariant-set condition in
Theorem~\ref{thm:lyapunov-degree-m} certifies exactness.

\begin{example}
\label{ex:curved-slow-set}
Consider
\[
\begin{aligned}
    \dot x_1
    &=
    -\bigl(x_1-(x_2^2+x_3^2)\bigr)
    -2(x_2^2+x_3^2)^{7/4},\\
    \dot x_2
    &=
    -\bigl(1+x_1-(x_2^2+x_3^2)\bigr)x_3
    -(x_2^2+x_3^2)^{3/4}x_2,\\
    \dot x_3
    &=
    \bigl(1+x_1-(x_2^2+x_3^2)\bigr)x_2
    -(x_2^2+x_3^2)^{3/4}x_3 .
\end{aligned}
\]
The curved set \(S=\{x\in\mathbb R^3:x_1=x_2^2+x_3^2\}\) carries the slow
algebraic dynamics, while transverse deviations decay exponentially.

Let \(r=\sqrt{x_2^2+x_3^2}\), and define \(y=x_1-r^2\). Since \(\dot r=-r^{5/2}\), one obtains
\(\dot y=\dot x_1-2r\dot r=-y\). Thus, in coordinates \((y,r,\theta)\), the
system satisfies \(\dot y=-y\), \(\dot r=-r^{5/2}\), and
\(\dot\theta=1+y\). Let \(V(x)=\frac12((x_1-r^2)^2+r^2)=\frac12(y^2+r^2)\).
Near the origin, \(V\) is positive definite and equivalent to \(\|x\|^2\).
Moreover, \(\dot V=y\dot y+r\dot r=-y^2-r^{7/2}\). On a sufficiently small
neighborhood, \(y^2\ge |y|^{7/2}\), so
\(-\dot V\ge |y|^{7/2}+r^{7/2}\ge C(y^2+r^2)^{7/4}=C(2V)^{7/4}\) for some
\(C>0\). Therefore \(\dot V\le-\mu V^{7/4}\). Since
\(7/4=1+m/2\) corresponds to \(m=3/2\),
\eqref{eq:direct-upper-dissipation} holds. Moreover, \(S=\{y=0\}\) is
invariant because \(\dot y=-y\). On \(S\),
\(V=\frac12r^2\) and \(\dot V=-r^{7/2}=-2^{7/4}V^{7/4}\).
Hence~\eqref{eq:lyapunov-degree-m-lower} holds on
\(S\setminus\{0\}\), so Theorem~\ref{thm:lyapunov-degree-m} gives
\(3/2=\min\mathcal D_u\), and the origin is uniformly degree-\(3/2\) stable.
\end{example}

The final example applies the same theorem to a nonautonomous system. It combines a time-varying rotation with componentwise cubic damping,
so the radial decay is not written as an exact scalar normal form but is
controlled by norm inequalities.

\begin{example}
\label{ex:time-varying}
Consider the time-varying system
\[
\begin{aligned}
    \dot x_1&=-(1+e^{-t})x_2-(2+e^{-t})x_1^3,\\
    \dot x_2&=(1+e^{-t})x_1-(2+e^{-t})x_2^3 .
\end{aligned}
\]
Let \(V(x)=\frac12\|x\|^2\). The time-varying rotation terms cancel in
\(\dot V\), giving \(\dot V=-(2+e^{-t})(x_1^4+x_2^4)\). Since
\(2+e^{-t}\ge2\) and
\(x_1^4+x_2^4\ge\frac12(x_1^2+x_2^2)^2=2V^2\), we have
\(\dot V\le-4V^2\), so \eqref{eq:direct-upper-dissipation} holds with
\(m=2\).

For exactness, fix \(\rho>0\) sufficiently small and set
\(S=B_\rho(0)\). Since \(\dot V\le0\), \(S\) is forward invariant.
Moreover,
\(\dot V\ge-3(x_1^4+x_2^4)\ge-3(x_1^2+x_2^2)^2=-12V^2\).
Thus~\eqref{eq:lyapunov-degree-m-lower} holds on \(S\setminus\{0\}\), so \(2=\min\mathcal D_u\) and the
origin is uniformly degree-\(2\) stable by Theorem~\ref{thm:lyapunov-degree-m}. The same rate is visible from
\(r=\|x\|\): for \(r>0\),
\(\dot r=-(2+e^{-t})(x_1^4+x_2^4)/r\), and
\(\frac12r^4\le x_1^4+x_2^4\le r^4\) implies
\(-3r^3\le\dot r\le-r^3\), so
\(r(t)\asymp(t-t_0)^{-1/2}\).
\end{example}

Figure~\ref{fig:numerical-decay-verification} numerically confirms the
decay rates predicted in Examples~\ref{ex:hopf}--\ref{ex:time-varying}.
On log--log axes, algebraic decay of order \(t^{-1/m}\) has asymptotic slope
\(-1/m\), while exponential decay is faster than any power law. The numerical trajectories closely match the predicted reference rates.
Small deviations are expected, since the predicted degree fixes the asymptotic
decay exponent, but not the trajectory's scaling constants. The shift \(1+t\)
only avoids the singularity of the logarithmic axis at \(t=0\).

\begin{figure}[t]
    \centering

    \begin{subfigure}[t]{0.235\textwidth}
        \centering
        \includegraphics[width=\linewidth]{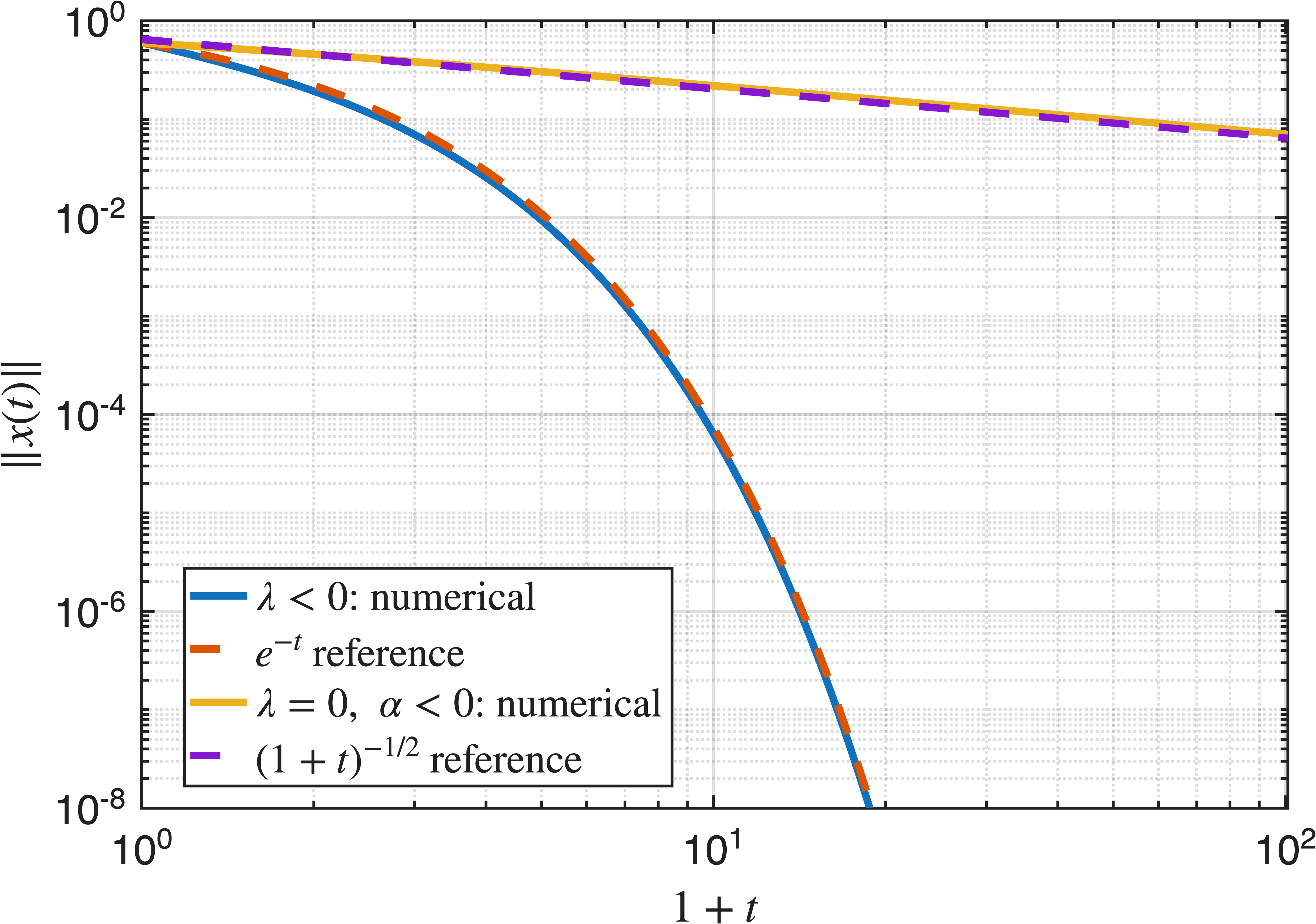}
        \caption{Example~\ref{ex:hopf}.}
        \label{fig:hopf-decay-2}
    \end{subfigure}
    \hfill
    \begin{subfigure}[t]{0.235\textwidth}
        \centering
        \includegraphics[width=\linewidth]{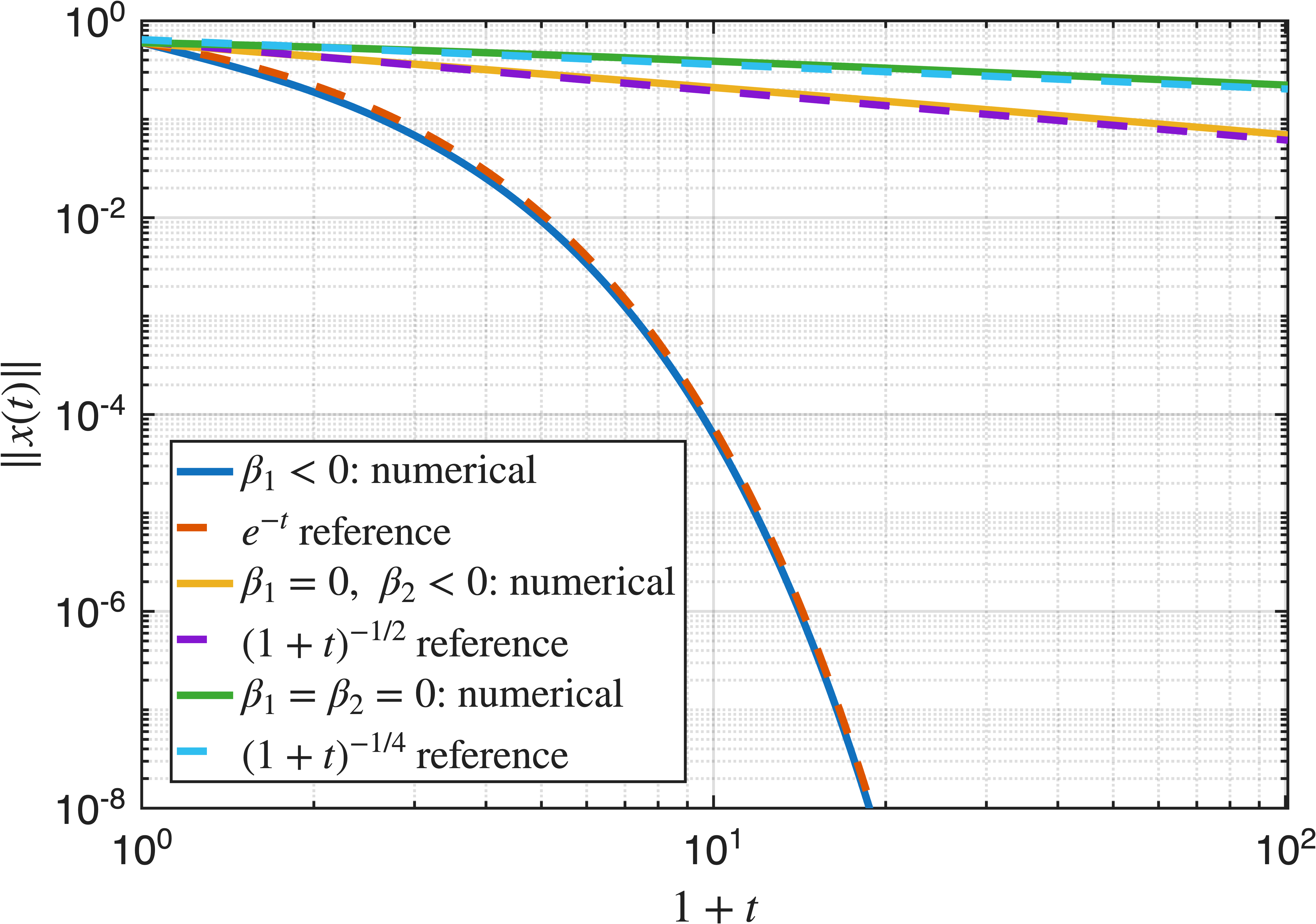}
        \caption{Example~\ref{ex:bautin}.}
        \label{fig:bautin-decay-2}
    \end{subfigure}
    \hfill
    \begin{subfigure}[t]{0.235\textwidth}
        \centering
        \includegraphics[width=\linewidth]{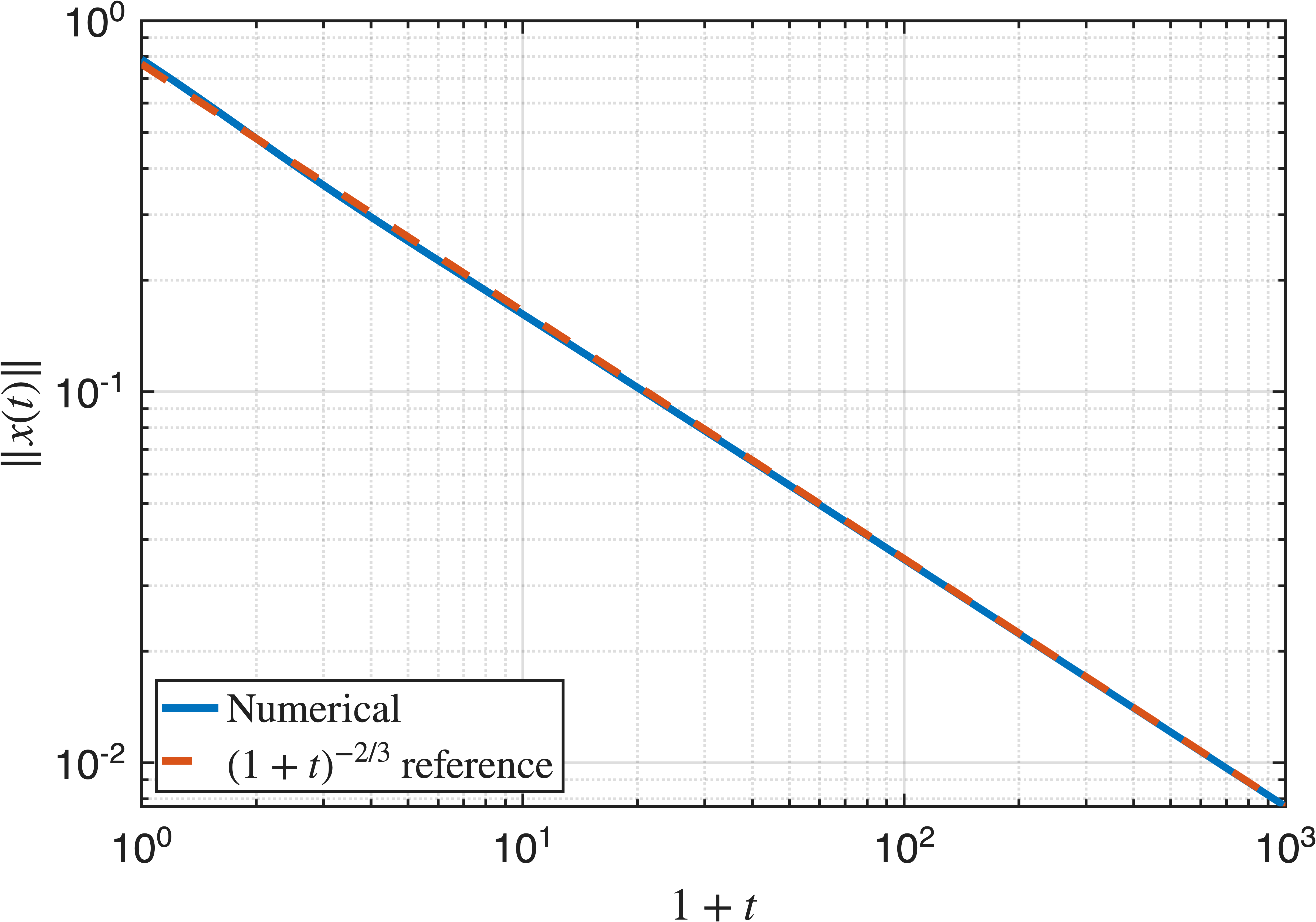}
        \caption{Example~\ref{ex:curved-slow-set}.}
        \label{fig:curved-slow-set-decay}
    \end{subfigure}
    \hfill
    \begin{subfigure}[t]{0.235\textwidth}
        \centering
        \includegraphics[width=\linewidth]{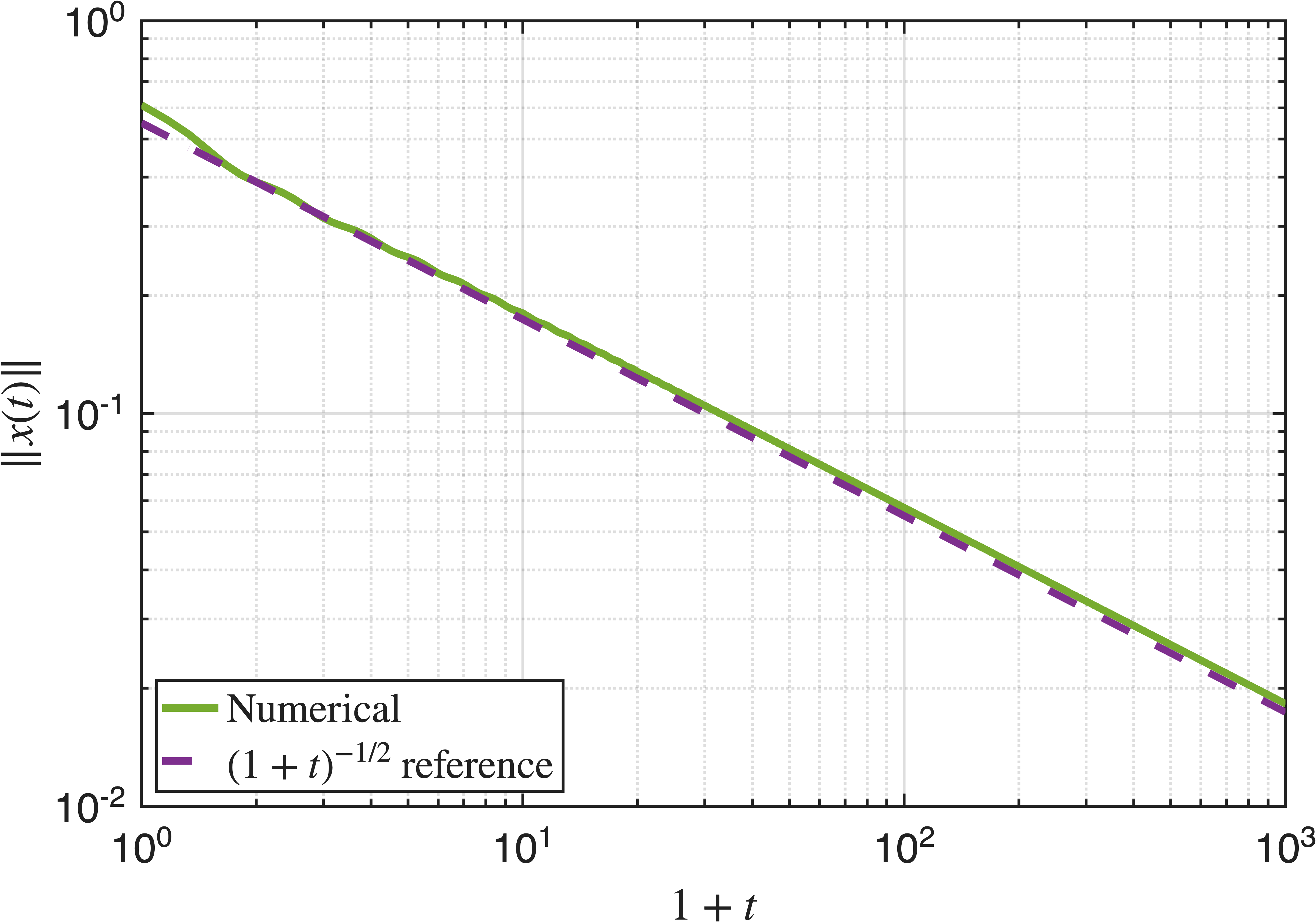}
        \caption{Example~\ref{ex:time-varying}.}
        \label{fig:time-varying-decay}
    \end{subfigure}

    \caption{
    Numerical verification of the decay rates certified in
    Examples~\ref{ex:hopf}--\ref{ex:time-varying}. Solid curves show
    \(\|x(t)\|\), while dashed curves show the predicted reference rates. The predicted and numerical rates are in near-exact agreement.}
    \label{fig:numerical-decay-verification}
\end{figure}

\subsection{Converse Lyapunov Characterization}

After the direct Lyapunov theorem, we state a converse result for admissible
degree-\(m\) decay. The converse Lyapunov function is obtained from the flow by a
supremum--integral construction. Although this construction yields a continuous
Lyapunov function, it need not yield a smooth one; hence the derivative
condition is stated using the upper Dini derivative along solutions. This result characterizes membership in
\(\mathcal D_u\), rather than exactness: exact degree stability still requires
showing that no smaller degree is admissible.

Throughout this subsection, let \(\Psi(t;t_0,x_0)\) denote the solution at
time \(t\) initialized from \(x_0\) at time \(t_0\). For continuous \(V:\mathbb{R}_{\ge0}\times B_r(0)\to\mathbb{R}_{\ge0}\),
define the upper Dini derivative along solutions by
\begin{equation}
\label{eq:dini-derivative}
    D_f^+V(t,x):=
    \limsup_{h\downarrow0}
    \frac{V(t+h,\Psi(t+h;t,x))-V(t,x)}{h}.
\end{equation}
If \(V\in C^1\), then
\(D_f^+V=\partial_tV+\nabla_xV\cdot f(t,x)\).



\begin{theorem}
\label{thm:time-varying-nonsmooth-admissible-degree-iff}
Consider system~\eqref{eq:system} and let \(m\in\mathbb R_{\ge0}\). Then \(m\in\mathcal D_u\), equivalently, the origin admits a uniform
degree-\(m\) estimate, if and only if there exist \(r>0\), a continuous
function
\(V:\mathbb R_{\ge0}\times B_r(0)\to\mathbb R_{\ge0}\), and constants
\(c_1,c_2,c_3>0\) such that, for all \(t\ge0\) and \(x\in B_r(0)\),
\begin{equation}
\label{eq:converse-quadratic-bounds}
    c_1\|x\|^2\le V(t,x)\le c_2\|x\|^2,
\end{equation}
and
\begin{equation}
\label{eq:converse-dini-dissipation}
    D_f^+V(t,x)
    \le
    -c_3V(t,x)^{1+\frac m2}.
\end{equation}
\end{theorem}

\begin{proofsketch}
For sufficiency, Theorem~\ref{thm:lyapunov-degree-m} applies with the
comparison lemma for upper Dini derivatives
(Lemma~3 in~\cite{moulay2008finite}) applied to
\eqref{eq:converse-dini-dissipation} in place of the comparison argument
for~\eqref{eq:direct-upper-dissipation}. Conversely, suppose \(m\in\mathcal D_u\), with uniform-estimate constants
\(\alpha,K,R>0\). Choose \(r>0\) such that \(Kr<R\), so that every trajectory
starting in \(B_r(0)\) remains in \(B_R(0)\). For \(s\ge0\), write
\(\Psi_s^{t,x}:=\Psi(t+s;t,x)\), fix \(a>0\), and define
\begin{equation}
\label{eq:converse-V-construction-1}
    V(t,x)
    :=
    \sup_{\tau\ge0}
    \left\{
        \|\Psi_\tau^{t,x}\|^2
        +
        a\int_0^\tau
        \|\Psi_s^{t,x}\|^{m+2}\,ds
    \right\}.
\end{equation}
By~\eqref{eq:degree-m-envelope} and~\eqref{eq:degree-m-estimate},
\(\|\Psi_s^{t,x}\|\le K\phi_m^\alpha(s;\|x\|)\). Hence
\(\|\Psi_\tau^{t,x}\|^2\le K^2\|x\|^2\), and direct integration of
\((\phi_m^\alpha(s;\|x\|))^{m+2}\) gives
\(\int_0^\tau\|\Psi_s^{t,x}\|^{m+2}\,ds
\le \frac{K^{m+2}}{2\alpha}\|x\|^2\) for all \(m\ge0\).
Since \(\tau=0\) in~\eqref{eq:converse-V-construction-1} gives \(V(t,x)\ge\|x\|^2\), it follows that
\(\|x\|^2\le V(t,x)\le\left(K^2+\frac{aK^{m+2}}{2\alpha}\right)\|x\|^2.
\) Thus \(V\) is finite and satisfies~\eqref{eq:converse-quadratic-bounds}.

We next establish continuity. At \(x=0\), it follows directly
from~\eqref{eq:converse-quadratic-bounds}. For each finite
\(T\), the supremum and integral over \([0,T]\) are continuous by continuous dependence of solutions. On any compact neighborhood \(U\) away from
\(x=0\), \eqref{eq:degree-m-envelope}--\eqref{eq:degree-m-estimate} give
\(\sup_{(t,x)\in U}\sup_{\tau\ge T}\|\Psi_\tau^{t,x}\|^2\to0\) and
\(\sup_{(t,x)\in U}\int_T^\infty\|\Psi_s^{t,x}\|^{m+2}ds\to0\)
as \(T\to\infty\). These bounds show that the portion with \(\tau\ge T\)
converges uniformly on \(U\) to
\(a\int_0^\infty\|\Psi_s^{t,x}\|^{m+2}ds\); hence \(V\) is continuous. The flow's semigroup property gives
\(
V(t+h,\Psi(t+h;t,x))
\le
V(t,x)
-a\int_0^h\|\Psi(t+s;t,x)\|^{m+2}\,ds.
\)
Dividing by \(h\) and letting \(h\downarrow0\) yields
\(
D_f^+V(t,x)
\le-a\|x\|^{m+2}
\le-c_3V(t,x)^{1+m/2},
\)
where
\(
c_3
=
a\left(
K^2+\frac{aK^{m+2}}{2\alpha}
\right)^{-(1+m/2)}.
\)
Thus~\eqref{eq:converse-dini-dissipation} holds, completing the necessity
direction. The detailed proof is provided in
Appendix~\ref{app:proof-converse}.
\end{proofsketch}

We next illustrate the converse construction~\eqref{eq:converse-V-construction-1} on the critical Hopf normal form from Example~\ref{ex:hopf}.

\begin{example}
\label{ex:converse-critical-hopf}
Consider the Hopf normal form in Example~\ref{ex:hopf} at the critical
parameters \(\lambda=0\), \(\alpha<0\). For a trajectory initialized at
\(x(0)=x\), let \(r(t):=\|x(t)\|\) and \(r_0:=r(0)=\|x\|\). Then
\(r(t)=r_0(1+2|\alpha|r_0^2t)^{-1/2}\), so \(2\in\mathcal D_u\). Using the construction~\eqref{eq:converse-V-construction-1}, define
\(V(x):=\sup_{\tau\ge0}\{r(\tau)^2+
a\int_0^\tau r(s)^4\,ds\}\), with \(a>0\). Since
\(\frac{d}{dt}r(t)^2=-2|\alpha|r(t)^4\), one has
\(\int_0^\tau r(s)^4\,ds=
(r_0^2-r(\tau)^2)/(2|\alpha|)\). Therefore the expression inside the
supremum equals
\(\frac{a}{2|\alpha|}r_0^2+
(1-\frac{a}{2|\alpha|})r(\tau)^2\). Since \(r(\tau)^2\) decreases from \(r_0^2\) to \(0\), the supremum is \(r_0^2\) when \(a\le2|\alpha|\) and \(\frac{a}{2|\alpha|}r_0^2\) when \(a>2|\alpha|\). Hence \(V(x)=C\|x\|^2\), where \(C:=\max\{1,a/(2|\alpha|)\}\). Consequently,
\(\dot V(x)=-2C|\alpha|\|x\|^4
=-(2|\alpha|/C)V(x)^2\).
Thus, the converse construction yields a quadratic Lyapunov function with
matching degree-\(2\) dissipation. 
\end{example}



The construction~\eqref{eq:converse-V-construction-1} also applies directly to the time-varying system in
Example~\ref{ex:time-varying}.

\begin{example}
\label{ex:converse-time-varying-planar-degree-2}
For the system in Example~\ref{ex:time-varying}, let
\(z(s):=\|\Psi(t+s;t,x)\|^2\). Differentiating along solutions gives
\(\dot z(s)=-2(2+e^{-(t+s)})
(x_1(t+s)^4+x_2(t+s)^4)\le-2z(s)^2\), where we used
\(x_1^4+x_2^4\ge\frac12(x_1^2+x_2^2)^2\).
For \(a\in(0,2]\), define the converse Lyapunov function from
\eqref{eq:converse-V-construction-1} as
\(W_a(t,x):=\sup_{\tau\ge0}\{z(\tau)+a\int_0^\tau z(s)^2\,ds\}\).
The expression inside the supremum is nonincreasing in \(\tau\), since
\(\dot z(\tau)+az(\tau)^2\le-(2-a)z(\tau)^2\le0\).
Hence its supremum is attained at \(\tau=0\), giving
\(W_a(t,x)=z(0)=\|x\|^2\). Moreover,
\(-6W_a^2\le\dot W_a\le-2W_a^2\), which gives the uniform degree-\(2\)
certificate and its matching lower bound.
\end{example}

\section{Conclusion and Future Work}
We introduced a degree-based framework for local asymptotic stability that
places exponential and algebraic convergence rates on a common scale. We
developed Lyapunov certificates for admissibility and exact stability degree,
together with a nonsmooth converse characterization of uniform admissible
degrees, and illustrated the framework on Hopf and Bautin normal forms, along with
fractional-degree and time-varying systems. Future work will
investigate degree computation from leading vector-field terms, smooth
converse Lyapunov certificates, and computational methods for certifying
stability degrees in nonlinear systems.


\bibliographystyle{IEEEtran}
\bibliography{references}

\onecolumn
\appendices
{

\setlength{\parindent}{0pt}
\setlength{\parskip}{0.2em}

\section{Proof of Theorem~\ref{thm:lyapunov-degree-m}}
\label{app:proof-direct}

\begin{proof}
Choose any \(\bar r\in(0,r)\), and set
\[
    r_0:=\bar r\sqrt{\frac{a}{b}}.
\]
Fix \(t_0\ge0\) and \(x_0\in B_{r_0}(0)\). Let
\[
    x(t):=x(t,t_0,x_0),
    \qquad
    v(t):=V(t,x(t)).
\]

We first verify that the trajectory remains in \(B_{\bar r}(0)\), and
hence in \(B_r(0)\), the region on which
\eqref{eq:direct-quadratic-bounds}--\eqref{eq:direct-upper-dissipation}
hold. As long as \(x(t)\in B_{\bar r}(0)\),
\eqref{eq:direct-upper-dissipation} gives \(\dot v(t)\le0\). Hence,
using the upper bound in \eqref{eq:direct-quadratic-bounds},
\[
    v(t)
    \le v(t_0)
    \le b\|x_0\|^2
    <br_0^2
    =a\bar r^2.
\]
If the trajectory were to leave \(B_{\bar r}(0)\), let \(t_e>t_0\) be its
first exit time. Then \(\|x(t_e)\|=\bar r<r\), and the lower bound in
\eqref{eq:direct-quadratic-bounds} gives
\[
    a\bar r^2
    \le V(t_e,x(t_e))
    =v(t_e)
    <a\bar r^2,
\]
a contradiction. Thus
\[
    x(t)\in B_{\bar r}(0),
    \qquad t\ge t_0.
\]
Since \(\bar r<r\), this implies
\begin{equation}
\label{eq:direct-forward-containment}
    x(t)\in B_r(0),
    \qquad t\ge t_0.
\end{equation}

Since \(V\) is continuously differentiable and \(x(t)\) solves
\eqref{eq:system}, the function \(v\) is absolutely continuous on
compact time intervals and, by
\eqref{eq:direct-upper-dissipation} and
\eqref{eq:direct-forward-containment},
\begin{equation}
\label{eq:direct-v-upper}
    \dot v(t)
    \le
    -\mu v(t)^{1+\frac m2}
\end{equation}
for almost every \(t\ge t_0\).

First consider \(m=0\). Then \eqref{eq:direct-v-upper} reduces to
\(\dot v\le-\mu v\). Multiplying by \(e^{\mu t}\) gives
\[
    \frac{d}{dt}\bigl(e^{\mu t}v(t)\bigr)
    =e^{\mu t}\bigl(\dot v(t)+\mu v(t)\bigr)\le0.
\]
Hence \(e^{\mu t}v(t)\le e^{\mu t_0}v(t_0)\), and therefore
\[
    v(t)\le v(t_0)e^{-\mu(t-t_0)}.
\]
Applying both inequalities in \eqref{eq:direct-quadratic-bounds} gives
\begin{equation}
\label{eq:direct-degree-zero-bound}
    \|x(t,t_0,x_0)\|
    \le
    \sqrt{\frac ba}\,
    \|x_0\|e^{-\frac{\mu}{2}(t-t_0)}.
\end{equation}
By the \(m=0\) case of \eqref{eq:degree-m-envelope},
\eqref{eq:direct-degree-zero-bound} is a uniform degree-\(0\)
estimate with
\[
    c=\sqrt{\frac ba},
    \qquad
    \alpha=\frac{\mu}{2}.
\]
Hence \(0\in\mathcal D_u\). Since \(\mathcal D_u\subseteq\mathbb R_{\ge0}\), it follows immediately that
\[
0=\min\mathcal D_u.
\]
Thus the origin is uniformly degree-\(0\) stable, and this case coincides
with the classical Lyapunov criterion for exponential stability.

Now suppose \(m>0\). If a trajectory reaches the origin at some finite
time \(t_*\), uniqueness for~\eqref{eq:system} implies that
\(x(t)=0\) for all \(t\ge t_*\), so the desired estimate holds trivially
thereafter. Otherwise, by~\eqref{eq:direct-quadratic-bounds},
\(v(t)>0\), and we may therefore work with negative powers of \(v(t)\). From \eqref{eq:direct-v-upper},
\[
    \frac{d}{dt}v(t)^{-m/2}
    =
    -\frac m2v(t)^{-1-\frac m2}\dot v(t)
    \ge
    \frac m2\mu.
\]
Integrating from \(t_0\) to \(t\) gives
\[
    v(t)^{-m/2}
    \ge
    v(t_0)^{-m/2}
    +\frac m2\mu(t-t_0).
\]
Since both sides are positive, taking the power \(-2/m\), which reverses
the inequality as \(-2/m < 0\), yields
\[
    v(t)
    \le
    \left(
        v(t_0)^{-m/2}
        +\frac m2\mu(t-t_0)
    \right)^{-2/m}.
\]
Factoring \(v(t_0)^{-m/2}\) from the parentheses gives
\[
    v(t)
    \le
    \left[
        v(t_0)^{-m/2}
        \left(
            1+\frac m2\mu v(t_0)^{m/2}(t-t_0)
        \right)
    \right]^{-2/m},
\]
and therefore
\begin{equation}
\label{eq:direct-v-upper-integrated}
    v(t)
    \le
    v(t_0)
    \left(
        1+\frac m2\mu v(t_0)^{m/2}(t-t_0)
    \right)^{-2/m}.
\end{equation}

Using~\eqref{eq:direct-quadratic-bounds}, we have
\[
    v(t_0)\le b\|x_0\|^2,
    \qquad
    v(t_0)^{m/2}\ge a^{m/2}\|x_0\|^m.
\]
Substituting these into~\eqref{eq:direct-v-upper-integrated} gives
\[
    v(t)
    \le
    b\|x_0\|^2
    \left(
        1+\frac m2\mu a^{m/2}
        \|x_0\|^m(t-t_0)
    \right)^{-2/m},
\]
where the lower bound on \(v(t_0)\) is used inside the parentheses because
the exponent \(-2/m\) is negative.

Using again the lower bound in
\eqref{eq:direct-quadratic-bounds}, now at time \(t\),
\[
    a\|x(t,t_0,x_0)\|^2\le v(t),
\]
and therefore
\[
    \|x(t,t_0,x_0)\|^2
    \le
    \frac ba\,\|x_0\|^2
    \left(
        1+\frac m2\mu a^{m/2}
        \|x_0\|^m(t-t_0)
    \right)^{-2/m}.
\]
Taking square roots yields
\begin{equation}
\label{eq:direct-upper-trajectory-estimate}
    \|x(t,t_0,x_0)\|
    \le
    \sqrt{\frac ba}\,\|x_0\|
    \left(
        1+\frac m2\mu a^{m/2}
        \|x_0\|^m(t-t_0)
    \right)^{-1/m}.
\end{equation}
Comparison of \eqref{eq:direct-upper-trajectory-estimate} with
\eqref{eq:degree-m-envelope} shows that the origin admits a uniform
degree-\(m\) estimate with
\[
    c=\sqrt{\frac ba},
    \qquad
    \alpha=\frac{\mu}{2}a^{m/2}.
\]
Thus \(m\in\mathcal D_u\), and consequently
\(\underline m_u\le m\).

We now prove the exactness statement. Assume \(m>0\), and let \(S\)
and \(\nu\) satisfy the additional hypotheses of
Theorem~\ref{thm:lyapunov-degree-m}. Fix \(x_0\in S\setminus\{0\}\). By forward invariance,
\(x(t)\in S\) for all \(t\ge t_0\), so
\eqref{eq:lyapunov-degree-m-lower} gives
\begin{equation}
\label{eq:direct-v-lower}
    \dot v(t)
    \ge
    -\nu v(t)^{1+\frac m2}.
\end{equation}
For \(v(t)>0\), multiplying~\eqref{eq:direct-v-lower} by
\(-\frac m2v(t)^{-1-\frac m2}<0\) reverses the inequality and gives
\[
    -\frac m2v(t)^{-1-\frac m2}\dot v(t)
    \le
    \frac m2\nu.
\]
By the chain rule, the left-hand side is precisely
\[
    \frac{d}{dt}v(t)^{-m/2}
    =
    -\frac m2v(t)^{-1-\frac m2}\dot v(t),
\]
and hence
\[
    \frac{d}{dt}v(t)^{-m/2}
    \le
    \frac m2\nu.
\]
Integrating both sides from \(t_0\) to \(t\) yields
\[
    v(t)^{-m/2}-v(t_0)^{-m/2}
    \le
    \frac m2\nu(t-t_0),
\]
or equivalently,
\[
    v(t)^{-m/2}
    \le
    v(t_0)^{-m/2}
    +\frac m2\nu(t-t_0).
\]
and hence
\begin{equation}
\label{eq:direct-v-lower-integrated}
    v(t)
    \ge
    v(t_0)
    \left(
        1+\frac m2\nu v(t_0)^{m/2}(t-t_0)
    \right)^{-2/m}.
\end{equation}

For fixed \(t\ge t_0\), define
\[
    g(y):=
    y
    \left(
        1+\frac m2\nu y^{m/2}(t-t_0)
    \right)^{-2/m},
    \qquad y\ge0.
\]
Differentiating gives
\[
    g'(y)
    =
    \left(
        1+\frac m2\nu y^{m/2}(t-t_0)
    \right)^{-2/m-1}
    >0,
\]
so \(g\) is increasing on \(\mathbb R_{\ge0}\).
Equation~\eqref{eq:direct-v-lower-integrated} can be written as
\[
    v(t)\ge g(v(t_0)).
\]
Since~\eqref{eq:direct-quadratic-bounds} gives
\(v(t_0)\ge a\|x_0\|^2\), monotonicity of \(g\) implies
\[
    g(v(t_0))
    \ge
    g(a\|x_0\|^2).
\]
Therefore,
\[
    v(t)
    \ge
    g(v(t_0))
    \ge
    g(a\|x_0\|^2)
    =
    a\|x_0\|^2
    \left(
        1+\frac m2\nu a^{m/2}
        \|x_0\|^m(t-t_0)
    \right)^{-2/m}.
\]
On the other hand, the upper bound in
\eqref{eq:direct-quadratic-bounds}, evaluated at time \(t\), gives
\[
    v(t)\le b\|x(t,t_0,x_0)\|^2.
\]
Combining the two inequalities yields
\[
    b\|x(t,t_0,x_0)\|^2
    \ge
    a\|x_0\|^2
    \left(
        1+\frac m2\nu a^{m/2}
        \|x_0\|^m(t-t_0)
    \right)^{-2/m}.
\]
Dividing by \(b\) and taking square roots gives
\begin{equation}
\label{eq:direct-lower-trajectory-estimate}
    \|x(t,t_0,x_0)\|
    \ge
    C_1\|x_0\|
    \left(
        1+C_2\|x_0\|^m(t-t_0)
    \right)^{-1/m},
\end{equation}
where
\begin{equation}
\label{eq:direct-C1-C2}
    C_1:=\sqrt{\frac ab},
    \qquad
    C_2:=\frac m2\nu a^{m/2}.
\end{equation}

We show that \eqref{eq:direct-lower-trajectory-estimate} excludes
every uniform degree-\(q\) estimate with \(q<m\). Suppose, toward a
contradiction, that \(q<m\) belongs to \(\mathcal D_u\). By
Definition~\ref{def:degree-m-estimate}, there exist
\(\alpha_q,c_q,r_q>0\), independent of \(t_0\), such that
\begin{equation}
\label{eq:direct-assumed-degree-q}
    \|x(t,t_0,x_0)\|
    \le
    c_q
    \phi_q^{\alpha_q}(t-t_0;\|x_0\|)
\end{equation}
whenever \(\|x_0\|<r_q\).

Since \(0\in\overline{S\setminus\{0\}}\), choose
\(x_0\in S\setminus\{0\}\) with
\(\|x_0\|<\min\{r,r_q\}\), and set
\[
    \rho:=\|x_0\|>0,
    \qquad
    T:=t-t_0.
\]

If \(q=0\), then \eqref{eq:degree-m-envelope} and
\eqref{eq:direct-assumed-degree-q} give
\[
    \|x(t,t_0,x_0)\|
    \le
    c_q\rho e^{-\alpha_qT}.
\]
Combining this with \eqref{eq:direct-lower-trajectory-estimate} gives
\begin{equation}
\label{eq:direct-q-zero-contradiction}
    C_1(1+C_2\rho^mT)^{-1/m}
    \le
    c_qe^{-\alpha_qT},
    \qquad T\ge0.
\end{equation}
Multiplying \eqref{eq:direct-q-zero-contradiction} by
\(e^{\alpha_qT}\) gives
\[
    C_1e^{\alpha_qT}(1+C_2\rho^mT)^{-1/m}
    \le c_q.
\]
The left-hand side tends to \(+\infty\) as \(T\to\infty\), since an
exponential dominates every polynomial. This contradicts
\eqref{eq:direct-q-zero-contradiction}. Hence \(0\notin\mathcal D_u\).

Now let \(0<q<m\). By
\eqref{eq:degree-m-envelope} and
\eqref{eq:direct-assumed-degree-q},
\[
    \|x(t,t_0,x_0)\|
    \le
    c_q\rho
    (1+q\alpha_q\rho^qT)^{-1/q}.
\]
Combining this inequality with
\eqref{eq:direct-lower-trajectory-estimate} and cancelling
\(\rho>0\) gives
\begin{equation}
\label{eq:degree-q-contradiction}
    \frac{C_1}{c_q}
    (1+q\alpha_q\rho^qT)^{1/q}
    (1+C_2\rho^mT)^{-1/m}
    \le1,
    \qquad T\ge0.
\end{equation}

To examine the left-hand side of
\eqref{eq:degree-q-contradiction}, factor \(T\) from the two
parenthetical terms:
\begin{equation}
\label{eq:degree-q-asymptotic}
\begin{aligned}
&(1+q\alpha_q\rho^qT)^{1/q}
(1+C_2\rho^mT)^{-1/m}
\\
&\quad=
T^{\frac1q-\frac1m}
\left(\frac1T+q\alpha_q\rho^q\right)^{1/q}
\left(\frac1T+C_2\rho^m\right)^{-1/m}.
\end{aligned}
\end{equation}
The last two factors in \eqref{eq:degree-q-asymptotic} converge to
\[
    (q\alpha_q\rho^q)^{1/q}
    (C_2\rho^m)^{-1/m}
    =
    \frac{(q\alpha_q)^{1/q}}{C_2^{1/m}}
    >0.
\]
Moreover, \(0<q<m\) implies
\[
    \frac1q-\frac1m
    =
    \frac{m-q}{qm}
    >0.
\]
Consequently, \eqref{eq:degree-q-asymptotic} implies that the
left-hand side of \eqref{eq:degree-q-contradiction} tends to
\(+\infty\) as \(T\to\infty\), contradicting
\eqref{eq:degree-q-contradiction}. Therefore no
\(q\in(0,m)\) belongs to \(\mathcal D_u\).

We have proved
\[
    \mathcal D_u\cap[0,m)=\emptyset.
\]
Since \eqref{eq:direct-upper-trajectory-estimate} already established
\(m\in\mathcal D_u\), it follows that
\[
    \underline m_u=m=\min\mathcal D_u.
\]
By Definition~\ref{def:vector-degree-m-stability}, the origin is
uniformly degree-\(m\) stable.
\end{proof}

\section{Proof of Theorem~\ref{thm:time-varying-nonsmooth-admissible-degree-iff}}
\label{app:proof-converse}

\begin{proof}
We prove both directions.

We first recall the comparison lemma for upper Dini derivatives used
in the sufficiency direction.

\begin{lemma}[Comparison Lemma {\cite[Lemma~3]{moulay2008finite}}]
\label{lem:dini-comparison}
Let \(J\subseteq\mathbb R\) be an interval, and suppose the scalar
differential equation
\[
    \dot z=f(z),\qquad z\in J,
\]
admits a global semiflow
\[
    \Phi:\mathbb R_{\ge0}\times J\to J,
\]
where \(f:J\to\mathbb R\) is continuous. If
\(g:[a,b)\to J\) is continuous and satisfies
\[
    D^+g(t)\le f(g(t)),
    \qquad t\in[a,b),
\]
then
\[
    g(t)\le \Phi(t-a,g(a)),
    \qquad t\in[a,b).
\]
\end{lemma}

A proof and more general versions of the comparison lemma can be found in \cite[Sec.~5.2]{kartsatos1980advanced}.

We first prove sufficiency. Assume that \(V\) and
\(c_1,c_2,c_3,r>0\) satisfy
\eqref{eq:converse-quadratic-bounds} and
\eqref{eq:converse-dini-dissipation}. Choose any \(\bar r\in(0,r)\), and set
\[
    r_0:=\bar r\sqrt{\frac{c_1}{c_2}}.
\]
Fix \(t_0\ge0\) and \(x_0\in B_{r_0}(0)\), and define
\[
    x(t):=\Psi(t;t_0,x_0),
    \qquad
    v(t):=V(t,x(t)).
\]
By the semigroup property and \eqref{eq:dini-derivative},
\[
D^+v(t)=D_f^+V(t,x(t)).
\]
Hence, by \eqref{eq:converse-dini-dissipation},
\begin{equation}
\label{eq:converse-v-dini}
    D^+v(t)
    \le
    -c_3v(t)^{1+\frac m2}
\end{equation}
for every \(t\ge t_0\) for which \(x(t)\in B_r(0)\).

Consider the scalar comparison equation
\begin{equation}
\label{eq:converse-comparison-ode}
    \dot z=-c_3z^{1+\frac m2},
    \qquad
    z(t_0)=v(t_0).
\end{equation}
Its solution is nonincreasing and exists for all \(t\ge t_0\).
Applying Lemma~\ref{lem:dini-comparison} to
\eqref{eq:converse-v-dini}--\eqref{eq:converse-comparison-ode} yields
\begin{equation}
\label{eq:converse-v-comparison}
    v(t)\le z(t)\le z(t_0)=v(t_0)
\end{equation}
as long as \(x(t)\in B_r(0)\).

We next show that this local comparison is valid for all forward
time. If the trajectory left \(B_{\bar r}(0)\), let \(t_e>t_0\) be its
first exit time. Then \(\|x(t_e)\|=\bar r<r\), whereas
\eqref{eq:converse-v-comparison} and the upper bound in
\eqref{eq:converse-quadratic-bounds} give
\[
    V(t_e,x(t_e))
    \le
    V(t_0,x_0)
    \le
    c_2\|x_0\|^2
    <
    c_2r_0^2
    =
    c_1\bar r^2.
\]
The lower bound in \eqref{eq:converse-quadratic-bounds}, however,
gives
\[
    V(t_e,x(t_e))
    \ge
    c_1\|x(t_e)\|^2
    =
    c_1\bar r^2,
\]
a contradiction. Thus
\[
    x(t)\in B_{\bar r}(0),
    \qquad t\ge t_0.
\]
Since \(\bar r<r\), this implies
\begin{equation}
\label{eq:converse-forward-containment}
    x(t)\in B_r(0),
    \qquad t\ge t_0,
\end{equation}
and \eqref{eq:converse-v-comparison} holds for every \(t\ge t_0\).

For \(m=0\), the solution of
\eqref{eq:converse-comparison-ode} is
\[
    z(t)=v(t_0)e^{-c_3(t-t_0)}.
\]
Hence, by \eqref{eq:converse-v-comparison},
\[
    v(t)\le v(t_0)e^{-c_3(t-t_0)}.
\]
Using the two inequalities in
\eqref{eq:converse-quadratic-bounds}, we obtain
\begin{equation}
\label{eq:converse-degree-zero-estimate}
    \|x(t)\|
    \le
    \sqrt{\frac{c_2}{c_1}}\,
    \|x_0\|
    e^{-\frac{c_3}{2}(t-t_0)}.
\end{equation}
By \eqref{eq:degree-m-envelope},
\eqref{eq:converse-degree-zero-estimate} is a uniform degree-\(0\)
estimate with
\[
    c=\sqrt{\frac{c_2}{c_1}},
    \qquad
    \alpha=\frac{c_3}{2}.
\]

Now suppose \(m>0\). Solving
\eqref{eq:converse-comparison-ode} gives
\[
    z(t)
    =
    v(t_0)
    \left(
        1+\frac m2c_3v(t_0)^{m/2}(t-t_0)
    \right)^{-2/m}.
\]
Together with \eqref{eq:converse-v-comparison},
\begin{equation}
\label{eq:converse-v-upper-integrated}
    v(t)
    \le
    v(t_0)
    \left(
        1+\frac m2c_3v(t_0)^{m/2}(t-t_0)
    \right)^{-2/m}.
\end{equation}

The upper bound in \eqref{eq:converse-quadratic-bounds} gives
\(v(t_0)\le c_2\|x_0\|^2\). Moreover, for fixed \(t\ge t_0\), define
\[
    h(y):=
    \left(
        1+\frac m2c_3y^{m/2}(t-t_0)
    \right)^{-2/m},
    \qquad y\ge0.
\]
For \(y>0\), differentiating gives
\[
    h'(y)
    =
    -c_3(t-t_0)y^{m/2-1}
    \left(
        1+\frac m2c_3y^{m/2}(t-t_0)
    \right)^{-2/m-1}
    \le0.
\]
Thus \(h\) is nonincreasing on \((0,\infty)\). Since \(h\) is continuous
at \(y=0\), it follows that \(h\) is nonincreasing on
\(\mathbb R_{\ge0}\).

By~\eqref{eq:converse-quadratic-bounds},
\[
    v(t_0)\ge c_1\|x_0\|^2.
\]
Since \(h\) is nonincreasing, this implies
\[
    h(v(t_0))
    \le
    h(c_1\|x_0\|^2),
\]
that is,
\[
    \left(
        1+\frac m2c_3v(t_0)^{m/2}(t-t_0)
    \right)^{-2/m}
    \le
    \left(
        1+\frac m2c_3c_1^{m/2}
        \|x_0\|^m(t-t_0)
    \right)^{-2/m}.
\]

Also, the upper bound in
\eqref{eq:converse-quadratic-bounds} gives
\[
    v(t_0)\le c_2\|x_0\|^2.
\]
Substituting both estimates into
\eqref{eq:converse-v-upper-integrated} yields
\[
    v(t)
    \le
    c_2\|x_0\|^2
    \left(
        1+\frac m2c_3c_1^{m/2}
        \|x_0\|^m(t-t_0)
    \right)^{-2/m}.
\]

Using now the lower bound in
\eqref{eq:converse-quadratic-bounds} at time \(t\),
\[
    c_1\|x(t)\|^2\le v(t),
\]
we obtain
\[
    c_1\|x(t)\|^2
    \le
    c_2\|x_0\|^2
    \left(
        1+\frac m2c_3c_1^{m/2}
        \|x_0\|^m(t-t_0)
    \right)^{-2/m}.
\]
Dividing by \(c_1\) and taking square roots gives
\begin{equation}
\label{eq:converse-upper-trajectory-estimate}
    \|x(t)\|
    \le
    \sqrt{\frac{c_2}{c_1}}\,
    \|x_0\|
    \left(
        1+\frac m2c_3c_1^{m/2}
        \|x_0\|^m(t-t_0)
    \right)^{-1/m}.
\end{equation}
Comparing \eqref{eq:converse-upper-trajectory-estimate} with
\eqref{eq:degree-m-envelope}, we may choose
\[
    c=\sqrt{\frac{c_2}{c_1}},
    \qquad
    \alpha=\frac{c_3}{2}c_1^{m/2}.
\]
Thus the origin admits a uniform degree-\(m\) estimate, so
\(m\in\mathcal D_u\).

We now prove necessity. Suppose \(m\in\mathcal D_u\). By
Definition~\ref{def:degree-m-estimate}, there exist
\(\alpha,K,R>0\) such that
\begin{equation}
\label{eq:converse-assumed-degree-estimate}
    \|\Psi(t;t_0,x_0)\|
    \le
    K\phi_m^\alpha(t-t_0;\|x_0\|)
\end{equation}
for every \(t_0\ge0\), \(x_0\in B_R(0)\), and \(t\ge t_0\).

Choose \(r>0\) such that \(Kr<R\). Since
\(\phi_m^\alpha(s;\rho)\le\rho\) for \(s,\rho\ge0\),
\eqref{eq:converse-assumed-degree-estimate} implies
\[
    \|\Psi(t;t_0,x_0)\|
    \le K\|x_0\|
    <R
\]
for every \(x_0\in B_r(0)\). Hence every such solution remains in
\(B_R(0)\) for all forward time.

For \(s\ge0\), write
\[
    \Psi_s^{t,x}:=\Psi(t+s;t,x),
\]
fix \(a>0\), and define
\begin{equation}
\label{eq:converse-V-construction}
    V(t,x)
    :=
    \sup_{\tau\ge0}
    \left\{
        \|\Psi_\tau^{t,x}\|^2
        +
        a\int_0^\tau
        \|\Psi_s^{t,x}\|^{m+2}\,ds
    \right\}.
\end{equation}

We first establish quadratic bounds for
\eqref{eq:converse-V-construction}. Taking \(\tau=0\) in~\eqref{eq:converse-V-construction}, the integral term
vanishes and \(\Psi_0^{t,x}=x\). Hence the quantity inside the supremum
equals \(\|x\|^2\), and therefore
\begin{equation}
\label{eq:converse-V-lower-bound}
    V(t,x)\ge\|x\|^2.
\end{equation}
Moreover, by~\eqref{eq:converse-assumed-degree-estimate},
\begin{equation}
\label{eq:converse-flow-degree-estimate}
    \|\Psi_\tau^{t,x}\|
    \le
    K\phi_m^\alpha(\tau;\|x\|),
    \qquad \tau\ge0.
\end{equation}
Since~\eqref{eq:degree-m-envelope} gives
\[
    \phi_m^\alpha(\tau;\|x\|)
    \le
    \phi_m^\alpha(0;\|x\|)
    =
    \|x\|,
    \qquad \tau\ge0,
\]
it follows that
\begin{equation}
\label{eq:bound1}
    \|\Psi_\tau^{t,x}\|^2
    \le
    K^2\|x\|^2,
    \qquad \tau\ge0.
\end{equation}

For \(m>0\), using~\eqref{eq:degree-m-envelope} in
\eqref{eq:converse-flow-degree-estimate} gives
\[
    \|\Psi_s^{t,x}\|
    \le
    K\|x\|
    \left(
        1+m\alpha\|x\|^ms
    \right)^{-1/m}.
\]
Raising both sides to the power \(m+2>0\) yields
\[
    \|\Psi_s^{t,x}\|^{m+2}
    \le
    K^{m+2}\|x\|^{m+2}
    \left(
        1+m\alpha\|x\|^ms
    \right)^{-\frac{m+2}{m}},
\]
and direct integration yields
\begin{equation}
\label{eq:converse-integral-bound-positive-m}
\begin{aligned}
    \int_0^\infty
    \|\Psi_s^{t,x}\|^{m+2}\,ds
    &\le
    K^{m+2}\|x\|^{m+2}
    \int_0^\infty
    \left(
        1+m\alpha\|x\|^ms
    \right)^{-\frac{m+2}{m}} ds \\
    &=
    \frac{K^{m+2}}{2\alpha}\|x\|^2.
\end{aligned}
\end{equation}
For \(m=0\), using~\eqref{eq:degree-m-envelope} in
\eqref{eq:converse-flow-degree-estimate} gives
\[
    \|\Psi_s^{t,x}\|
    \le K\|x\|e^{-\alpha s},
\]
and therefore
\begin{equation}
\label{eq:converse-integral-bound-zero}
\begin{aligned}
    \int_0^\infty
    \|\Psi_s^{t,x}\|^2\,ds
    &\le
    K^2\|x\|^2
    \int_0^\infty e^{-2\alpha s}\,ds \\
    &=
    \frac{K^2}{2\alpha}\|x\|^2.
\end{aligned}
\end{equation}
Since the right-hand side of
\eqref{eq:converse-integral-bound-positive-m} reduces to
\(\frac{K^2}{2\alpha}\|x\|^2\) when \(m=0\),
\eqref{eq:converse-integral-bound-positive-m} in fact holds for all
\(m\ge0\).

Since the integrands are nonnegative, for every \(\tau\ge0\),
\[
    \int_0^\tau
    \|\Psi_s^{t,x}\|^{m+2}\,ds
    \le
    \int_0^\infty
    \|\Psi_s^{t,x}\|^{m+2}\,ds.
\]
Hence, by~\eqref{eq:converse-integral-bound-positive-m} for \(m \geq 0\),
\begin{equation}
\label{eq:bound2}
 \int_0^\tau
    \|\Psi_s^{t,x}\|^{m+2}\,ds
    \le
    \frac{K^{m+2}}{2\alpha}\|x\|^2,
    \qquad \tau\ge0.   
\end{equation}

Combining \eqref{eq:bound1} and \eqref{eq:bound2} with the definition
\eqref{eq:converse-V-construction}, for every \(\tau\ge0\) we have
\[
    \|\Psi_\tau^{t,x}\|^2
    +
    a\int_0^\tau\|\Psi_s^{t,x}\|^{m+2}\,ds
    \le
    \left(
        K^2+\frac{aK^{m+2}}{2\alpha}
    \right)\|x\|^2.
\]
Taking the supremum over \(\tau\ge0\) and using
\eqref{eq:converse-V-construction} therefore gives
\[
    V(t,x)
    \le
    \left(
        K^2+\frac{aK^{m+2}}{2\alpha}
    \right)\|x\|^2.
\]
Together with the lower bound~\eqref{eq:converse-V-lower-bound}, we obtain
\begin{equation}
\label{eq:converse-constructed-quadratic-bounds}
    \|x\|^2
    \le
    V(t,x)
    \le
    c_2\|x\|^2,
    \qquad
    t\ge0,\quad x\in B_r(0),
\end{equation}
where
\[
    c_2:=
    K^2+\frac{aK^{m+2}}{2\alpha}>0.
\]
In particular, \(c_2\) is independent of \(t\).

We next prove continuity of \(V\). At \(x=0\),
\eqref{eq:converse-constructed-quadratic-bounds} gives
\[
    0\le V(t,x)\le c_2\|x\|^2\to0
    \qquad\text{as }x\to0,
\]
uniformly in \(t\).

Fix
\((t_\ast,x_\ast)\in\mathbb R_{\ge0}\times B_r(0)\) with
\(x_\ast\neq0\), and choose a compact neighborhood \(U\) of
\((t_\ast,x_\ast)\) such that
\[
    \|x\|\ge\delta>0,
    \qquad (t,x)\in U.
\]
For \(T>0\), define
\begin{equation}
\label{eq:converse-F-definition}
    F(\tau,t,x)
    :=
    \|\Psi_\tau^{t,x}\|^2
    +
    a\int_0^\tau
    \|\Psi_s^{t,x}\|^{m+2}\,ds,
    \qquad
    (\tau,t,x)\in[0,T]\times U.
\end{equation}
Continuous dependence of solutions on initial time and initial state
implies that the first term in
\eqref{eq:converse-F-definition} is continuous. The integral term is
also continuous. Indeed, if
\((\tau_k,t_k,x_k)\to(\tau,t,x)\), then
\[
    \int_0^{\tau_k}
    \|\Psi_s^{t_k,x_k}\|^{m+2}\,ds
    =
    \int_0^T
    \mathbf 1_{[0,\tau_k]}(s)
    \|\Psi_s^{t_k,x_k}\|^{m+2}\,ds.
\]
The integrands converge pointwise for all \(s\neq\tau\) and are
uniformly bounded on the relevant compact set. Dominated convergence
therefore shows that \(F\) is continuous on \([0,T]\times U\).

Define
\begin{equation}
\label{eq:converse-VT-definition}
    V_T(t,x)
    :=
    \max_{0\le\tau\le T}F(\tau,t,x).
\end{equation}
Since \(F\) is continuous and \([0,T]\) is compact, more precisely, the constant correspondence \((t,x)\mapsto[0,T]\) is continuous with nonempty compact values, Berge's Maximum Theorem~\cite[Theorem~17.31]{aliprantis2006infinite}
implies that \(V_T\) is continuous on \(U\).

To pass to the infinite horizon, define
\begin{equation}
\label{eq:converse-I-definition}
    I(t,x)
    :=
    \int_0^\infty
    \|\Psi_s^{t,x}\|^{m+2}\,ds.
\end{equation}
The bounds obtained from
\eqref{eq:converse-assumed-degree-estimate} and
\eqref{eq:degree-m-envelope}, together with
\(\|x\|\ge\delta\) on \(U\), imply
\begin{equation}
\label{eq:converse-integral-tail}
    \sup_{(t,x)\in U}
    \int_T^\infty
    \|\Psi_s^{t,x}\|^{m+2}\,ds
    \longrightarrow0
\end{equation}
and
\begin{equation}
\label{eq:converse-state-tail}
    \sup_{(t,x)\in U}
    \sup_{\tau\ge T}
    \|\Psi_\tau^{t,x}\|^2
    \longrightarrow0
\end{equation}
as \(T\to\infty\). In particular,
\eqref{eq:converse-integral-tail} and continuity of the finite-horizon
integrals imply that \(I\) in \eqref{eq:converse-I-definition} is
continuous on \(U\).

For \(\tau\ge T\), equations
\eqref{eq:converse-F-definition} and
\eqref{eq:converse-I-definition} give
\begin{equation}
\label{eq:converse-F-tail-decomposition}
    F(\tau,t,x)
    =
    aI(t,x)
    +
    \|\Psi_\tau^{t,x}\|^2
    -
    a\int_\tau^\infty
    \|\Psi_s^{t,x}\|^{m+2}\,ds.
\end{equation}
By \eqref{eq:converse-integral-tail} and
\eqref{eq:converse-state-tail}, the last two terms in
\eqref{eq:converse-F-tail-decomposition} converge uniformly to zero
on \(U\). Hence
\begin{equation}
\label{eq:converse-tail-uniform}
    \sup_{(t,x)\in U}
    \left|
        \sup_{\tau\ge T}F(\tau,t,x)-aI(t,x)
    \right|
    \longrightarrow0.
\end{equation}

By \eqref{eq:converse-V-construction} and
\eqref{eq:converse-VT-definition},
\[
    V(t,x)
    =
    \max\left\{
        V_T(t,x),
        \sup_{\tau\ge T}F(\tau,t,x)
    \right\}.
\]
Since the maximum is Lipschitz in each argument,
\eqref{eq:converse-tail-uniform} implies
\[
    \sup_{(t,x)\in U}
    \left|
        V(t,x)
        -
        \max\{V_T(t,x),aI(t,x)\}
    \right|
    \longrightarrow0.
\]
For each \(T>0\), both \(V_T\) and \(I\) are continuous on \(U\).
Thus \(V\) is a uniform limit on \(U\) of continuous functions and is
therefore continuous there. Since
\((t_\ast,x_\ast)\) was arbitrary, and continuity at \(x=0\) follows
from \eqref{eq:converse-constructed-quadratic-bounds}, \(V\) is
continuous on \(\mathbb R_{\ge0}\times B_r(0)\).

It remains to establish the Dini derivative inequality. Fix \(x\in B_r(0)\). For all sufficiently small \(h>0\),
\(\Psi(t+h;t,x)\in B_r(0)\), and the semigroup property gives
\begin{equation}
\label{eq:converse-semigroup}
    \Psi(t+h+\tau;t+h,\Psi(t+h;t,x))
    =
    \Psi(t+h+\tau;t,x).
\end{equation}
Applying \eqref{eq:converse-semigroup} to
\eqref{eq:converse-V-construction} yields
\[
\begin{aligned}
&V(t+h,\Psi(t+h;t,x))
\\
&=
\sup_{\tau\ge0}
\left\{
    \|\Psi(t+h+\tau;t,x)\|^2
    +
    a\int_h^{h+\tau}
    \|\Psi(t+s;t,x)\|^{m+2}\,ds
\right\}.
\end{aligned}
\]
With \(T=h+\tau\),
\[
\begin{aligned}
&V(t+h,\Psi(t+h;t,x))
\\
&=
\sup_{T\ge h}
\left\{
    \|\Psi(t+T;t,x)\|^2
    +
    a\int_0^T
    \|\Psi(t+s;t,x)\|^{m+2}\,ds
    -
    a\int_0^h
    \|\Psi(t+s;t,x)\|^{m+2}\,ds
\right\}.
\end{aligned}
\]
Since the final integral is independent of \(T\), and the supremum
over \(T\ge h\) is bounded by that over \(T\ge0\), definition
\eqref{eq:converse-V-construction} gives
\begin{equation}
\label{eq:converse-dynamic-programming}
    V(t+h,\Psi(t+h;t,x))
    \le
    V(t,x)
    -
    a\int_0^h
    \|\Psi(t+s;t,x)\|^{m+2}\,ds.
\end{equation}
Subtracting \(V(t,x)\) from both sides of
\eqref{eq:converse-dynamic-programming}, dividing by \(h>0\), and
taking \(\limsup_{h\downarrow0}\), we obtain from
\eqref{eq:dini-derivative}
\begin{equation}
\label{eq:converse-dini-state}
    D_f^+V(t,x)
    \le
    -a\|x\|^{m+2}.
\end{equation}

Finally, the upper bound in
\eqref{eq:converse-constructed-quadratic-bounds} implies
\[
    \|x\|^{m+2}
    \ge
    c_2^{-\left(1+\frac m2\right)}
    V(t,x)^{1+\frac m2}.
\]
Substituting this into \eqref{eq:converse-dini-state} gives
\begin{equation}
\label{eq:converse-final-dissipation}
    D_f^+V(t,x)
    \le
    -a c_2^{-\left(1+\frac m2\right)}
    V(t,x)^{1+\frac m2}.
\end{equation}
Thus \eqref{eq:converse-quadratic-bounds} and
\eqref{eq:converse-dini-dissipation} hold for the constructed
Lyapunov function with
\[
    c_1=1,
    \qquad
    c_3
    =
    a c_2^{-\left(1+\frac m2\right)}.
\]
This proves the converse implication and completes the proof.
\end{proof}
}

\end{document}